\documentclass[11pt,a4paper]{scrartcl}
\usepackage{a4wide}
\usepackage[UKenglish]{babel}
\usepackage{latexsym}
\usepackage{amssymb}
\usepackage{mathrsfs}
\usepackage{graphicx}
\usepackage[latin1]{inputenc}
\usepackage{amsmath}
\usepackage{amsfonts}
\usepackage{amsthm}
\usepackage{url}
\usepackage{mathptmx}
\usepackage{helvet}

\usepackage{eurosym}
\usepackage{hyperref} 
\hypersetup{pdftitle={Option pricing in multivariate stochastic volatility models of OU type
},pdfauthor={Oliver Pfaffel},colorlinks=true,linkcolor=black,citecolor=black,urlcolor=black,breaklinks=true}
\usepackage{pdfsync}

\newtheorem{theorem}{Theorem}[section]

\newtheorem{proposition}[theorem]{Proposition}
\newtheorem{lemma}[theorem]{Lemma}
\newtheorem{definition}[theorem]{Definition}

\newtheorem{example}{Example}[section]
\newtheorem{remark}[theorem]{Remark}

\newcommand{\C}{\mathbb{C}}

\newcommand{\N}{\mathbb{N}}
\newcommand{\R}{\mathbb{R}}
\newcommand{\Z}{\mathbb{Z}}
\renewcommand{\L}{\mathrm{L}}
\renewcommand{\S}{\mathbb{S}}

\newcommand{\E}{E}
\newcommand{\T}{\mathsf{T}}
\newcommand{\tr}{\mathrm{tr}}
\renewcommand{\vec}{{\mathrm{vec}}}
\newcommand{\vech}{{\mathrm{vech}}}
\renewcommand{\Re}{{\mathrm{Re}}}
\renewcommand{\Im}{{\mathrm{Im}}}
\newcommand{\beq}{\begin{equation}}
\newcommand{\eeq}{\end{equation}}
\newcommand{\beqa}{\begin{eqnarray*}}
\newcommand{\eeqa}{\end{eqnarray*}\noindent}
\newcommand{\beqan}{\begin{eqnarray}}
\newcommand{\eeqan}{\end{eqnarray}\noindent}
\newcommand{\ba}{\begin{array}}
\newcommand{\ea}{\end{array}}
\newcommand{\bpm}{\begin{pmatrix}}
\newcommand{\epm}{\end{pmatrix}}
\newcommand{\bsm}{\begin{smallmatrix}}
\newcommand{\esm}{\end{smallmatrix}}
\newcommand{\levy}{L\'{e}vy{ }}
\renewcommand{\sp}[2]{\left\langle #1,#2 \right\rangle}
\newcommand{\norm}[1]{\left|\left|#1\right|\right|}
\newcommand{\Fcal}{\mathcal{F}}
\newcommand{\Ecal}{\mathcal{E}}
\newcommand{\Bcal}{\mathcal{B}}
\newcommand{\Abf}{\mathbf{A}}
\newcommand{\ra}{\rightarrow}
\newcommand{\bs}{\backslash}
\renewcommand{\(}{\left(}
\renewcommand{\)}{\right)}
\newcommand{\eqd}{\stackrel{fidi}{=}}
\newcommand{\Bscr}{\mathscr{B}}
\newcommand{\Cscr}{\mathscr{C}}
\newcommand{\wh}{\widehat}
\newcommand{\wt}{\widetilde}
\newcommand{\1}{\mathbf{1}}
\newcommand{\W}{\mathcal{W}}
\newcommand{\eur}{\textnormal{\euro}}

\numberwithin{equation}{section}
\title{Option Pricing in Multivariate Stochastic Volatility Models of OU Type\thanks{The first author gratefully acknowledges support from the FWF (Austrian Science Fund) under grant P19456 and by the National Centre of Competence in Research ``Financial Valuation and Risk Management'' (NCCR FINRISK), Project D1 (Mathematical Methods in Financial Risk Management) of the Swiss National Science Foundation. The second and third author greatly appreciate the support of the Technische Universit\"at M\"unchen - Institute of Advanced Study, the second author additionally the one by the International Graduate School of Science and Engineering, funded by the German Excellence Initiative. The authors thank Christa Cucchiero for fruitful discussions. They are also grateful to two anonymous referees and an anonymous associate editor for their numerous helpful comments, that significantly improved the present article.}} 

\author{Johannes Muhle-Karbe\thanks{Departement f\"ur Mathematik, ETH Z\"urich, R\"amistrasse 101 CH-8092 Z\"urich, Switzerland. \emph{Email:} \ttfamily{johannes.muhle-karbe@math.ethz.ch}} \and Oliver Pfaffel\thanks{TUM Institute for Advanced Study \& Zentrum Mathematik, Technische Universit\"at München, Parkring 13, D-85748 Garching, Germany. \emph{Email:} \ttfamily{pfaffel@ma.tum.de}} \and Robert Stelzer\thanks{Institute of Mathematical Finance, Ulm University, Helmholtzstraße 18, D-89081 Ulm, Germany. \emph{Email:} \ttfamily{robert.stelzer@uni-ulm.de}}}

\date{November 1, 2011}

\begin{document}
\maketitle
\newcommand{\slugmaster}{%
\slugger{MMedia}{xxxx}{xx}{x}{x--x}}

\begin{abstract}
We present a multivariate stochastic volatility model with leverage, which is flexible enough to recapture the individual dynamics as well as the interdependencies between several assets, while still being highly analytically tractable.

First, we derive the characteristic function and give conditions that ensure its analyticity and absolute integrability in some open complex strip around zero. Therefore we can use Fourier methods to compute the prices of multi-asset options efficiently. To show the applicability of our results, we propose a concrete specification, the OU-Wishart model, where the dynamics of each individual asset coincide with the popular $\Gamma$-OU BNS model. This model can be well calibrated to market prices, which we illustrate with an example using options on the exchange rates of some major currencies. Finally, we show that covariance swaps can also be priced in closed form.
\end{abstract}

\begin{tabbing}
\end{tabbing}

\vspace{0.15cm}\noindent\emph{Keywords:}  multivariate stochastic volatility models, OU-type processes,  option pricing

\pagestyle{myheadings}
\thispagestyle{plain}
\markboth{}{Option Pricing in Multivariate Stochastic Volatility Models of OU Type}
\newpage 

\section{Introduction}

This paper deals with the pricing of options depending on several underlying assets. While there is a vast amount of literature on the pricing of single-asset options, see, e.g., \cite{Cont2004,Schoutens2003} for an overview, the amount of literature considering the multi-asset case is rather limited. This is most likely due to the fact that the trade-off between \emph{flexibility} and \emph{tractability} is particularly delicate in a multivariate setting. On the one hand, the model under consideration should be flexible enough to recapture stylized facts observed in real option prices. When dealing with multiple underlyings, this becomes challenging, since not only the individual assets but also their joint behaviour has to be taken into account.  On the other hand, one needs enough mathematical structure to calculate option prices in the first place and to be able to calibrate the model to market prices. Due to an increasing number of state variables and parameters, this is also not an easy task in a multidimensional framework. In this article we propose the multivariate OU-type stochastic volatility model of Pigorsch and Stelzer~\cite{pigorsch} in the generalised form introduced by Barndorff-Nielsen and Stelzer~\cite{Barndorff-Nielsen2009}, which seems to present a reasonable compromise between these competing requirements. 

The log-price processes $Y=(Y^1,\ldots,Y^d)$ of $d$ financial assets are modelled as
\beqan
d Y_t &=& (\mu + \beta(\Sigma_t) )\,dt + \Sigma_t^\frac12 \, dW_t + \rho(dL_t), \label{intro_Y}\\\
d\Sigma_t &=& (A\Sigma_t + \Sigma_t A^\T) \,dt + dL_t , \label{intro_Sigma}
\eeqan
where $\mu\in\R^d$, $A$ is a real $d\times d$ matrix, and $\beta$, $\rho$ are linear operators from the real $d\times d$ matrices to $\R^d$. Moreover, $W$ is an $\R^d$-valued Wiener process and  $L$ is an independent matrix subordinator, i.e., a \levy process which only has positive semidefinite increments. Hence, the covariance process $\Sigma$ is an Ornstein-Uhlenbeck (henceforth OU) type process with values in the positive semidefinite matrices, cf. Barndorff-Nielsen and Stelzer \cite{barndorff}. Thus we call (\ref{intro_Y}), (\ref{intro_Sigma}) the \emph{multivariate stochastic volatility model of OU type}. The positive semidefinite OU type process $\Sigma$ introduces a stochastic volatility and, what is difficult to achieve using several univariate models, a stochastic correlation between the assets. Moreover, $\Sigma$ is mean reverting and increases only by jumps. The jumps represent the arrival of new information that results in positive shocks in the volatility and positive or negative shocks in the correlation of some assets. Due to the leverage term $\rho(dL_t)$ they are correlated with price jumps. The present model is a multivariate generalisation of the non-Gaussian OU type stochastic volatility model introduced by Barndorff-Nielsen and Shephard \cite{Barndorff-Nielsen2001} (henceforth BNS model). For one underlying, these models are found to be both flexible and tractable in  Nicolato and Venardos \cite{Nicolato2003}. The key reason is that  the characteristic function of the return process can often be computed in closed form, which allows European options to be be priced efficiently using the Fourier methods introduced by Carr and Madan \cite{Carr1999} as well as Raible \cite{raible.2000}. In the present study, we show that a similar approach is also applicable in the multivariate case. Recently, Benth and Vos \cite{Benth} discussed a somewhat similar model in the context of energy markets. However, they do not establish conditions for the applicability of Fourier pricing and, more importantly, do not calibrate their model to market prices.

Alternatively, the covariance process $\Sigma$ can also be modelled by other processes taking values in the positive semidefinite matrices. In particular, several authors have advocated to use a diffusion model based on the Wishart process, cf., e.g., Da Fonseca, Grasselli, and Tebaldi \cite{fonseca}, Gourieroux \cite{gou}, Gourieroux and Sufana \cite{Gourieroux2010}, and Da Fonseca and Grasselli~\cite{dafonseca.grasselli.10}. This leads to a multivariate generalisation of the model of Heston \cite{Heston1993}. However, there is empirical evidence suggesting that volatility jumps (together with the stock price), cf. Jacod and Todorov \cite{Jacod2009}, which cannot be recaptured by a diffusion model. Moreover, the treatment of square-root processes on the cone of positive semidefinite matrices is mathematically quite involved, see Cuchiero, Filipovi\'c, Mayerhofer, and Teichmann\cite{Cuchiero2009}.\footnote{This study generalizes the theory of \emph{affine} processes from the positive univariate factors treated in \cite{DPS2000,dfs2003} to factor processes taking values in the cone of symmetric positive semidefinite matrices. In particular, to ensure the existence of square-root processes, a quite intricate drift condition turns out to be necessary.} For example, whereas Da Fonseca and Grasselli \cite{dafonseca.grasselli.10} have very recently succeeded in calibrating their model to market prices, the resulting parameters do no satisfy the drift condition for the existence of the underlying square-root diffusion, suggesting that a more sophisticated optimization routine is necessary. 

Another possible approach is to consider multivariate models based on a concatenation of univariate building blocks. This approach is taken, e.g., by Luciano and Schoutens \cite{Luciano2006} using L\'evy processes, by Dimitroff, Lorenz, and Szimayer \cite{Dimitroff2009}, who consider a multivariate Heston model, and by Hubalek and Nicolato \cite{Hubalek}, who put forward a multifactor BNS model. However, all these models either have a somewhat limited capability to catch complex dependence structures (compare Section \ref{ss:empirical}) or lead to tricky (factor) identification issues. Apart from models where all parameters are determined by single-asset options, we are not aware of successful calibrations of such models. The paper of Ma \cite{Ma2009} proposes a two-dimensional Black-Scholes model where the correlation between the two Brownian motions is stochastic and given by a diffusion process with values in an interval contained in $[-1,1]$. However, pricing can only be done via Monte-Carlo simulation in this model. In addition, an extension to higher dimensions is not obvious, since the necessary positive semidefiniteness of the correlation matrix of the Brownian noise imposes additional constraints, which are hard to incorporate.
  
The remainder of this paper is organised as follows. Sections~\ref{sec: OU} and \ref{subsection: Definition} introduce the multivariate stochastic volatility model of OU type. Afterwards, we derive the joint characteristic function of $(Y_t,\Sigma_t)$. We then show in Section~\ref{sec: mgf} that a simple moment condition on $L$ implies analyticity and absolute integrability of the moment generating function of $Y_t$ in some open complex strip around zero. Equivalent martingale measures are discussed in Section~\ref{sec: emm}, where we also present a subclass that preserves the structure of our model. In Section~\ref{pricing}, we recall how to use Fourier methods to compute prices of multi-asset options efficiently. Subsequently, we propose the OU-Wishart model, where $L$ is a compound Poisson process with Wishart distributed jumps. It turns out that the OU-Wishart model has margins which are in distribution equivalent to a $\Gamma$-OU BNS model, one of the tractable specifications commonly used in the univariate case. Moreover, the characteristic function can be computed in closed form, which makes option pricing and calibration particularly feasible. In an illustrative example we calibrate a bivariate OU-Wishart model to market prices, and compare its performance to the multivariate Variance Gamma model of \cite{Luciano2006} and a multivariate extension with stochastic volatility. As a final application, we show in Section~\ref{sec: cov swaps} that covariance swaps can also be priced in closed form. The appendix contains a result on multidimensional analytic functions which is needed to establish the regularity of the moment generating function in Section~\ref{sec: mgf}.

\subsection*{Notation}

$M_{d,n}(\R)$ (resp.\ $M_{d,n}(\C)$) represent the $d\times n$ matrices with real (resp.\ complex) entries. We abbreviate $M_d(\cdot)=M_{d,d}(\cdot)$. $\S_d$ denotes the subspace of $M_d(\R)$ of all symmetric matrices. We write $\S_d^+$ for the cone of all positive semidefinite matrices, and $\S_d^{++}$ for the open cone of all positive definite matrices. The identity matrix in $M_d(\R)$ is denoted by $I_d$. $\sigma(A)$ denotes the set of all eigenvalues of $A\in M_d(\C)$. We write $\Re(z)$ and $\Im(z)$ for the real or imaginary part of $z\in\C^d$ or $z\in M_d(\C)$, which has to be understood componentwise. The components of a vector or matrix are denoted by subscripts, however for stochastic processes we use superscripts to avoid double indices.

On $\R^d$, we typically use the Euclidean scalar product, $\sp{x}{y}_{\R^d}:=x^\T y$, and on $M_d(\R)$ or $\S_d$ the scalar products given by $\sp{A}{B}_{M_d(\R)}:=\tr(A^\T B)$ or $\sp{A}{B}_{\S_d}:=\tr(A B)$, respectively. However, due to the equivalence of all norms on finite dimensional vector spaces, most results hold independently of the norm. We also write $\sp{x}{y}=x^\T y$ for $x,y\in\C^d$, although this is only a bilinear form but not a scalar product on $\C^d$.

We denote by $\vec:M_d(\R)\to\R^{d^2}$ the bijective linear operator that stacks the columns of a matrix below one other. With the above norms, $\vec$ is a Hilbert space isometry. Likewise, for a symmetric matrix $S\in\S_d$ we denote by $\vech(S)$ the vector consisting of the columns of the upper-diagonal part including the diagonal.

Furthermore, we employ an intuitive notation concerning integration with respect to matrix-valued processes. For an $M_{m,n}(\R)$-valued \levy process $L$, and $M_{d,m}(\R)$ resp.\ $M_{n,p}(\R)$- valued  processes $X,Y$ integrable with respect to $L$, the term $\int_0^t X_s \, dL_s Y_s$ is to be understood as the $d\times p$ (random) matrix with $(i,j)$-th entry  $\sum_{k=1}^m \sum_{l=1}^n \int_0^t X_s^{ik} \, dL_s^{kl} Y_s^{lj}$.

\section{The multivariate stochastic volatility model of OU type}\label{s:model}

For the remainder of the paper, fix a filtered probability space $(\Omega,\Fcal,(\Fcal_t)_{t\in [0,T]},P)$ in the sense of \cite[Definition I.1.3]{Jacod2003}, where $\Fcal_0=\{\Omega,\emptyset\}$ is trivial and $T>0$ is a a fixed terminal time.

\subsection{Positive semidefinite processes of OU type}\label{sec: OU}

To formulate our model, we need to introduce the concept of matrix subordinators as studied in \cite{bnpa08}. 

\begin{definition}
An $\S_d$-valued \levy Process $L=(L_t)_{t\in\R_+}$ is called \emph{matrix subordinator}, if $L_t-L_s\in\S_d^+$ for all $t>s$.
\end{definition}

The characteristic function of a matrix subordinator $L$ is given by $\E(e^{i\tr(Z L_1)}) =\exp(\psi_L(Z))$ for the \emph{characteristic exponent}
\[ \psi_L(Z) = i\tr(\gamma_L Z) + \int_{\S_d^+} ( e^{i\tr(XZ)} - 1 ) \, \kappa_L(dX),\quad Z \in M_d(\R), \]
where $\gamma_L\in\S_d^+$ and $\kappa_L$ is a \levy measure on $\S_d$ satisfying $\kappa_L(\S_d\bs\S_d^+)=0$ as well as $\int_{\{\norm{X}\leq 1\}} \norm{X} \,\kappa_L(dX) < \infty$.

\emph{Positive semidefinite processes of OU type} are a generalisation of nonnegative OU type processes (cf.\ \cite{barndorff}). Let $L$ be a matrix subordinator and $A\in M_d(\R)$. The positive semidefinite OU type process $\Sigma=(\Sigma_t)_{t\in\R_+}$ is defined as the unique strong solution to the stochastic differential equation
\beq d\Sigma_t = (A\Sigma_t + \Sigma_t A^\T) \,dt + dL_t, \quad \Sigma_0\in\S_d^+. \label{sde_ou_general}\eeq
It is given by
\beq \Sigma_t = e^{A t} \Sigma_0 e^{A^\T t} + \int_0^t e^{A(t-s)} \, dL_s \, e^{A^\T(t-s)}. \label{ou:solution}\eeq
Since $\Sigma_t\in\S_d^{+}$ for all $t\in\R_+$, this process can be used to model the stochastic evolution of a covariance matrix. As in the univariate case there exists a closed form expression for the integrated volatility. Suppose 
\beq 0 \notin \sigma(A)+\sigma(A). \label{cond:eigenvalues}\eeq
Then, the integrated OU type process $\Sigma^+$ is given by
\beq \Sigma_t^+ := \int_0^t \Sigma_s \,ds = \Abf^{-1}(\Sigma_t-\Sigma_0-L_t), \label{prop_int_ou}\eeq
where $\Abf:X\mapsto A X + X A^\T$. Note that condition (\ref{cond:eigenvalues}) implies that the operator $\Abf$ is invertible, cf. \cite[Theorem 4.4.5]{horn90}. In the case where $\Sigma$ is \emph{mean reverting}, i.e., $A$ only has eigenvalues with strictly negative real part, condition (\ref{cond:eigenvalues}) is trivially satisfied.

\subsection{Definition and marginal dynamics of the model}\label{subsection: Definition}

The following model was introduced and studied in \cite{pigorsch}  from a statistical point of view in the no-leverage case and has also been considered in \cite{Barndorff-Nielsen2009}. Here we discuss its applicability to option pricing.

 Let $L$ be a matrix subordinator with characteristic exponent $\psi_L$ and $W$ an independent $\R^d$-valued Wiener process. The \emph{multivariate stochastic volatility model of OU type} is then given by
\beqan
d Y_t &=& (\mu + \beta(\Sigma_t) )\,dt + \Sigma_t^\frac12 \, dW_t + \rho(dL_t), \quad Y_0\in\R^d  \label{sde_Y}\\
d\Sigma_t &=& (A\Sigma_t + \Sigma_t A^\T) \,dt + dL_t, \quad \Sigma_0\in\S_d^+, \label{sde_Sigma}
\eeqan
with linear operators $\beta,\rho:M_d(\R) \ra \R^d$, $\mu\in\R^d$, and $A\in M_d(\R)$ such that $0 \notin \sigma(A)+\sigma(A)$.

We have specified the \emph{risk premium} $\beta$ and the \emph{leverage operator} $\rho$ in a quite general form. The following specification turns out to be particularly tractable.

\begin{definition}
We call $\beta$ and $\rho$ \emph{diagonal} if, for $\beta_1,\ldots,\beta_d\in\R$ and $\rho_1,\ldots,\rho_d \in\R$,
\[ \beta(X) = \( \ba{c} \beta_1 X_{11} \\ \vdots \\ \beta_d X_{dd} \ea \), \quad \rho(X) = \( \ba{c} \rho_1 X_{11} \\ \vdots \\ \rho_d X_{dd} \ea \), \quad\forall\, X\in M_d(\R). \]
\end{definition}

In the following, we will denote for each $i \in \{1,\ldots,d\}$ by $\beta^i(X)$ and $\rho^i(X)$ the $i$-th component of the vector $\beta(X)$ or $\rho(X)$, respectively. The marginal dynamics of the individual assets have been derived in \cite[Proposition 4.3]{Barndorff-Nielsen2009}:

\begin{theorem}
Let $i \in \{1,\ldots,d\}$. Then we have
\[
\(Y_t^i\)_{t\in\R_+} \eqd \( \mu_{i}t + \beta^i(\Sigma_t^+) + \int_0^t (\Sigma_s^{ii})^{\frac12} \,dW_s^{i} + \rho^i(L_t) \)_{t\in\R_+},
\]
where $\eqd$ denotes equality of all finite dimensional distributions.
\label{theorem_marginal_dynamics}\end{theorem}

Let us now consider the case where $A$ is a diagonal matrix, $A=\left(\bsm a_1 & & 0 \\ & \ddots & \\ 0 & & a_d \esm\right)$, and $\beta$, $\rho$ are diagonal as well. Then, for every $i \in \{1,\ldots,d\}$, we have
\beqan
dY_t^i &\eqd & (\mu_i + \beta_i \Sigma_t^{ii}) \,dt + (\Sigma_t^{ii})^{\frac12} \,dW_t^i + \rho_i \,dL_t^{ii}, \label{sde: Y marginal}\\
d\Sigma_t^{ii} &=& 2a_i\Sigma_t^{ii} \,dt + dL_t^{ii} \label{sde: Sigma marginal}.
\eeqan
Evidently, every diagonal element $L^{ii}$, $i=1,\ldots,d$, of a matrix subordinator $L$ is a univariate subordinator, and thus $\Sigma^{ii}$ is a nonnegative OU type process. Consequently, the model for the $i$-th asset is equivalent in distribution to a univariate BNS model.

\subsection{Characteristic function}\label{sec: cf}

Let $\sp{\cdot}{\cdot}_V$, $\sp{\cdot}{\cdot}_W$ be bilinear forms as introduced in the notation, where $V$,$W$ may be either $\R^d$, $\C^d$ or $M_d(\cdot)$. Given a linear operator $T:V\to W$, the \emph{adjoint} $T^*:W\to V$ is the unique linear operator such that $\sp{Tx}{y}_W=\sp{x}{T^* y}_V$ for all $x\in V$ and $y\in W$. Directly by definition we obtain the following:

\begin{lemma}\label{adjointforfourier}
Let $y\in\R^d$, $\,z\in M_d(\R)$ and $\,t\in\R_+$. Then the adjoints of the linear operators 
\begin{align*}
\Abf &: X\mapsto A X + X A^\T, \quad \Bscr(t):X\mapsto e^{A t} X e^{A^\T t} - X, \\
\Cscr(t)&:X\mapsto  e^{At}Xe^{A^\T t} z + \beta(\Abf^{-1}(\Bscr(t)X))y^\T + \rho(X)y^\T + \frac{i}{2} y y^\T \Abf^{-1}(\Bscr(t)X)
\end{align*}
on $M_d(\C)$ are given by
\begin{align*}
\Abf^*&:X\mapsto A^\T X + X A, \quad \Bscr(t)^*:X\mapsto e^{A^\T t} X e^{A t} - X, \\
\Cscr(t)^*&:X\mapsto e^{A^\T t} X z^\T e^{A t} + \rho^*(Xy) + \Bscr(t)^*\Abf^{-*} \left( \beta^*(Xy)+\frac{i}{2} X y y^\T \right).
\end{align*}
\end{lemma}

Note that for diagonal $\rho$ it holds that $\rho^*(X)=\left(\bsm \rho_1 X_{11}  &   & 0 \\
																														 & \ddots &   \\
																														0 & & \rho_d X_{dd}  \esm\right)$ for all $X\in M_d(\R)$.
																														
Our main objective in this section is to compute the joint characteristic function of $(Y_t,\Sigma_t)$. This will pave the way for Fourier pricing of multi-asset options later on.  Note that we use the scalar product 
\[ \sp{(x_1,y_1)}{(x_2,y_2)} := x_1^\T x_2 + \tr(y_1^\T y_2) \]
on $\R^d\times M_d(\R)$.

\begin{theorem}[Joint characteristic function]
For every $(y,z)\in\R^d\times M_d(\R)$ and $t\in\R_+$, the joint characteristic function of $(Y_t,\Sigma_t)$ is given by
\begin{align*}
&\E[ \exp \(i\sp{(y,z)}{(Y_t,\Sigma_t)}\) ] \\
&= \exp\left\{ i y^\T(Y_0+\mu t) + i \tr( \Sigma_0 e^{A^\T t} z e^{A t} ) \right. \\
& \left. \quad + i \tr\left(\Sigma_0\left(e^{A^\T t} \Abf^{-*} \left( \beta^*(y)+\frac{i}{2} y y^\T \right) e^{A t} - \Abf^{-*} \left( \beta^*(y)+\frac{i}{2} y y^\T \right) \right)\right) \right. \\
& \left. \quad + \int_0^t \psi_L \left(  e^{A^\T s} z e^{A s} + \rho^*(y) + e^{A^\T s} \Abf^{-*} \left( \beta^*(y)+\frac{i}{2} y y^\T \right) e^{A s} - \Abf^{-*} \left( \beta^*(y)+\frac{i}{2} y y^\T \right)\right)
ds\right\},
\end{align*}
where $\Abf^{-*}:=(\Abf^{*})^{-1}$ denotes the inverse of the adjoint of $\Abf:X\mapsto AX+XA^\T$, that is, the inverse of $\Abf^*:X\mapsto A^\T X + X A$.
\label{cf}\end{theorem}

Note that for $z=0$ we obtain the characteristic function of $Y_t$.

\begin{proof}
Since $\Sigma$ is adapted to the filtration generated by $L$, and by the independence of $L$ and $W$,
\begin{align*}
  \E[ \exp\(\sp{(y,z)}{(Y_t,\Sigma_t)}\)& ] = e^{i y^\T(Y_0+\mu t)} \E\left[e^{i\tr(z^\T \Sigma_t) + i y^\T(\beta(\Sigma_t^+)+\rho(L_t))} \E\left(e^{i y^\T\int_0^t \Sigma_s^{\frac{1}{2}}\,dW_s } \Big| (L_s)_{s\in\R_+} \right) \right] \\
&  = e^{i y^\T(Y_0+\mu t)} \E\left[e^{i\tr(z^\T \Sigma_t) + i y^\T(\beta(\Sigma_t^+)+\rho(L_t))} \exp\left(-\frac12 y^\T \Sigma_t^+ y\right) \right].
\end{align*}
By (\ref{prop_int_ou}) and using the fact that the trace is invariant under cyclic permutations the last term equals
\[
e^{i y^\T(Y_0+\mu t)} \E\left[e^{i\tr(z^\T \Sigma_t + \beta(\Abf^{-1}(\Sigma_t-\Sigma_0-L_t)) y^\T  + \rho(L_t) y^\T + \frac{i}{2} y y^\T \Abf^{-1}(\Sigma_t-\Sigma_0-L_t) )} \right].
\]
In view of (\ref{ou:solution}), we have
\[ \Sigma_t - \Sigma_0 - L_t = \int_0^t \Bscr(t-s)\,dL_s + \Bscr(t)\Sigma_0, \]
for the linear operator $ \Bscr(t)$ from Lemma \ref{adjointforfourier}. Therefore,
\begin{align*}
&\E[ \exp(i \sp{(y,z)}{(Y_t,\Sigma_t)} ] \\
&= \exp\(i y^\T(Y_0+\mu t) + i\tr\(z^\T e^{A t} \Sigma_0 e^{A^\T t} + \beta(\Abf^{-1}(\Bscr(t)\Sigma_0)) y^\T + \frac{i}{2} y y^\T \Abf^{-1}(\Bscr(t)\Sigma_0)\)\) \\
& \quad \times\, 
\E\left[\exp\left(i\tr\left(z^\T \int_0^t e^{A(t-s)}\,dL_s\,e^{A^\T(t-s)} + \beta\left(\Abf^{-1}\left(\int_0^t e^{A(t-s)}\,dL_s\,e^{A^\T(t-s)} - L_t\right)\right) y^\T  \right.\right.\right. \\
& \left.\left.\left. \qquad\qquad +\, \rho(L_t) y^\T + \frac{i}{2} y y^\T \Abf^{-1}\left(\int_0^t e^{A(t-s)}\,dL_s\,e^{A^\T(t-s)} - L_t \right)\right)\right) \right] \\
&= \exp\left(i y^\T(Y_0+\mu t) + i \tr\left(z^\T e^{A t} \Sigma_0 e^{A^\T t} + \beta(\Abf^{-1}(\Bscr(t)\Sigma_0)) y^\T + \frac{i}{2} y y^\T \Abf^{-1}(\Bscr(t)\Sigma_0)\right)\right) \\
& \quad \times\, 
\E\left[\exp\left(i\tr\left(\left(\int_0^t \Cscr(t-s)\,dL_s \right)^\T I_d \right)\right) \right]
\end{align*}
with the linear operator $ \Cscr(t)$ from Lemma \ref{adjointforfourier}, since $\Abf^{-1}\left(\int_0^t e^{A(t-s)}\,dL_s\,e^{A^\T(t-s)} - L_t\right)\in\S_d$. An immediate multivariate generalisation of results obtained in \cite[Proposition 2.4]{rajput89} (see also \cite[Lemma 3.1]{Eberlein1999}) yields an explicit formula for the expectation above:
\[ \E\left[\exp\left(i\tr\left(\left(\int_0^t \Cscr(t-s)\,dL_s \right)^\T I_d \right)\right) \right] = \exp\left( \int_0^t \psi_L \left( \Cscr(s)^* I_d\right) \,ds \right).
\]
By Lemma \ref{adjointforfourier} we have
\[ e^{\int_0^t \psi_L \left( \Cscr(s)^* I_d\right) \,ds } = 
 e^{ \int_0^t \psi_L \left(  e^{A^\T s} z^\T e^{A s} + \rho^*(y) + e^{A^\T s} \Abf^{-*} \left( \beta^*(y)+\frac{i}{2} y y^\T \right) e^{A s} - \Abf^{-*} \left( \beta^*(y)+\frac{i}{2} y y^\T \right)  \right) \,ds }. \]
This expression is well-defined, because
\[ e^{A^\T s} z^\T e^{A s} + \rho^*(y) + e^{A^\T s} \Abf^{-*} \left( \beta^*(y)+\frac{i}{2} y y^\T \right) e^{A s} - \Abf^{-*} \left( \beta^*(y)+\frac{i}{2} y y^\T \right) \in M_d(\R)+i\S_d^+, \]
for all $s\in[0,t]$. Indeed, this follows from
\beq e^{A^\T s} \Abf^{-*} \left( y y^\T \right) e^{A s} - \Abf^{-*} \left(  y y^\T \right) = \int_0^s e^{A^\T u} y y^\T e^{A u} \,du \in\S_d^+. \label{leb_int_1}\eeq
Finally, we infer from Lemma \ref{adjointforfourier} that
\[
\tr\left(\beta(\Abf^{-1}(\Bscr(t)\Sigma_0)) y^\T + \frac{i}{2} y y^\T \Abf^{-1}(\Bscr(t)\Sigma_0)\right) 
= \tr\left(\Sigma_0\left(\Bscr(t)^*\Abf^{-*}\left(\beta^*(y) + \frac{i}{2} y y^\T \right)\right)\right),
\]
which gives the desired result by noting that $\tr(z \Sigma_t)=\tr(z^\T \Sigma_t)$.
\end{proof}

\subsection{Regularity of the moment generating function}\label{sec: mgf}

In this section we provide conditions ensuring that the characteristic function of $Y_t$ admits an analytic extension $\Phi_{Y_t}$ to some open convex neighbourhood of 0 in $\C^d$. Afterwards, we show absolute integrability. The regularity results obtained in this section will allow us to apply Fourier methods in Section \ref{pricing} to compute option prices efficiently.

\begin{definition}
For any $t\in [0,T]$, the \emph{moment generating function} of $Y_t$ is defined as
\[ \Phi_{Y_t}(y) := \E[ \exp(y^\T Y_t) ], \]
for all $y\in\C^d$ such that the expectation exists.
\end{definition}

Note that $\Phi_{Y_t}$ may not exist anywhere but on $i\R^d$, where it coincides with the characteristic function of $Y_t$. The next lemma is a first step towards conditions for the existence and analyticity of the moment generating function $\Phi_{Y_t}$ in a complex neighbourhood of zero.

\begin{lemma}\label{lemma_theta_L_analytic}
Let $L$ be a matrix subordinator with \emph{cumulant transform} $\Theta_L$, that is 
\[ \Theta_L(Z) = \psi_L(-iZ) = \tr(\gamma_L Z) + \int_{\S_d^+} (e^{\tr(XZ)} - 1) \,\kappa_L(dX), \quad Z\in M_d(\C), \]
and let $\epsilon>0$. Then $\Theta_L$ is analytic on the open convex set
\beq S_\epsilon := \{Z\in M_d(\C): \norm{\Re(Z)} < \epsilon \}-\S_d^+, \label{stip_theta_ana}\eeq
if and only if
\beq 
\int_{\{\norm{X}\geq 1\}} e^{\tr(RX)} \, \kappa_L(dX) < \infty \quad\textit{ for all } R\in M_d(\R) \textit{ with } \norm{R}<\epsilon. 
\label{lemma_finite_exp_moment}
\eeq
\end{lemma}

\begin{proof}
If (\ref{lemma_finite_exp_moment}) holds, \cite[Lemma A.2]{dfs2003} implies that $Z \mapsto  \E(e^{\tr(Z L_1)}) = e^{\Theta_L(Z)}$ is analytic on $S_\epsilon$. Due to Assumption \eqref{lemma_finite_exp_moment}, dominated convergence yields that $\Theta_L$ is continuous on $S_\epsilon$. The claim now follows from Lemma \ref{lemma_analytic_log}. Conversely, if $\Theta_L$ is analytic on $S_\epsilon$, then \cite[Lemma A.4]{dfs2003} implies that $\E(^{\tr(ZL_1)})=e^{\Theta_L(Z)}$ for all $Z\in S_\epsilon$. Thus, by \cite[Theorem 25.17]{sato99}, Condition (\ref{lemma_finite_exp_moment}) holds.
\end{proof}

The next theorem is a nontrivial (especially due to the involved heavy matrix calculus) generalization of \cite[Theorem 2.2]{Nicolato2003} to the multivariate case. It holds for all sub-multiplicative matrix norms on $M_d(\R)$ that satisfy $\norm{yy^\T}=\norm{y}^2$ for all $y\in\R^d$, where we use the Euclidean norm on $\R^d$. For example, this holds true for the Frobenius and the spectral norm (the operator norm associated to the Euclidean norm).

\begin{theorem}[Strip of analyticity]\label{theorem_analytic_cf}
Suppose the matrix subordinator $L$ satisfies
\beq 
\int_{\{\norm{X}\geq 1\}} e^{\tr(RX)} \, \kappa_L(dX) < \infty \quad\textit{ for all } R\in M_d(\R) \textit{ with } \norm{R} < \epsilon, 
\label{finite_exp_moment}
\eeq
for some $\epsilon>0$. Then the moment generating function $\Phi_{Y_t}$ of $Y_t$ is analytic on the open convex set
\[ S_{\theta} := \{y\in\C^d: \norm{\Re(y)}<\theta \}, \] 
where
\beq \theta := - \frac{\norm{\rho}}{(e^{2\norm{A}t}+1)\norm{\Abf^{-1}}} - \norm{\beta} + \sqrt{\Delta} > 0 \label{thetahat}\eeq
with
\[ \Delta := \left( \frac{\norm{\rho}}{(e^{2\norm{A}t}+1)\norm{\Abf^{-1}}} + \norm{\beta} \right)^2 + \frac{2\epsilon}{(e^{2\norm{A}t}+1)\norm{\Abf^{-1}}}. \]
Moreover, 
\beq \Phi_{Y_{t}}(y) = \exp\left( y^\T(Y_0+\mu t) + \tr(\Sigma_0 H_y(t)) + \int_0^t \Theta_L \left(H_y(s)+\rho^*(y) \right) \,ds \right) \label{eq: Phi}\eeq
for all $y\in S_\theta$, where 
\beq H_y(s) := e^{A^\T s} \Abf^{-*} \left( \beta^*(y)+\frac{1}{2} y y^\T \right) e^{A s} - \Abf^{-*} \left( \beta^*(y)+\frac{1}{2} y y^\T \right). \label{eq: H}\eeq
\end{theorem}

\begin{proof}
The main part of the proof is to show that the function
\[ G(y) := \exp\left( y^\T(Y_0+\mu t) + \tr(\Sigma_0 H_y(t)) + \int_0^t \Theta_L \left(H_y(s)+\rho^*(y) \right) \,ds \right) \]
is analytic on $S_{\theta}$. First we want to find a $\theta$ such that for all $u\in\R^d$ with $\norm{u}<\theta$, it holds that $\norm{H_u(s)+\rho^*(u)}<\epsilon$ for all $s\in[0,t]$. Since
\begin{align*}
\norm{H_u(s)+\rho^*(u)} &= \norm{e^{A^\T s} \Abf^{-*} \left( \beta^*(u)+\frac{1}{2} u u^\T \right) e^{A s} - \Abf^{-*} \left( \beta^*(u)+\frac{1}{2} u u^\T \right) + \rho^*(u) } \\
&\leq \frac12 (e^{2\norm{A}t}+1) \norm{\Abf^{-1}} \norm{u}^2 + \left( \norm{\rho} + (e^{2\norm{A}t}+1) \norm{\Abf^{-1}} \norm{\beta} \right) \norm{u},
\end{align*}
we have to find the roots of the polynomial
\[ p(x):= \frac12 (e^{2\norm{A}t}+1) \norm{\Abf^{-1}} x^2 + \left( \norm{\rho} + (e^{2\norm{A}t}+1) \norm{\Abf^{-1}} \norm{\beta} \right) x - \epsilon. \]
The positive one is given by $\theta$ as stated in (\ref{thetahat}). Note that $\theta>0$, because $p$ is a cup-shaped parabola with $p(0)=-\epsilon<0$. 

Now let $y\in S_{\theta}$, i.e., $y=u+iv$ with $\norm{u}<\theta$. Using $\Re(yy^\T)=uu^\T-vv^\T$ and (\ref{leb_int_1}) we get
\begin{align*}
\Re(H_y(s)+\rho^*(y)) &= H_u(s)+\rho^*(u) - \frac12 \left( e^{A^\T s}\Abf^{-*}(vv^\T)e^{As} - \Abf^{-*}(vv^\T) \right) \\
 &= H_u(s)+\rho^*(u) - \frac12 \int_0^s e^{A^\T r} vv^\T e^{Ar} \, dr.
\end{align*}
Because of $\int_0^s e^{A^\T r} vv^\T e^{Ar} \, dr \in \S_d^+$, we have
\begin{align*}
&\int_{\{ \norm{X}\geq 1\}} e^{\tr\(\Re(H_y(s)+\rho^*(y))X\)} \, \kappa_L(dX)\\
 &\qquad = \int_{\{\norm{X}\geq 1\}} e^{\tr((H_u(s)+\rho^*(u))X)}e^{- \frac12\tr\(\(\int_0^s e^{A^\T r} vv^\T e^{Ar} \, dr\)X\)} \, \kappa_L(dX) <\infty
\end{align*}
by Assumption (\ref{finite_exp_moment}), since $\norm{H_u(s)+\rho^*(u)}<\epsilon$. Thus, by Lemma \ref{lemma_theta_L_analytic} the function
\[ S_{\theta}\in y \mapsto \Theta_L(H_y(s)+\rho^*(y))\]
is analytic on $S_{\theta}$ for every $s\in[0,t]$. An application of Fubini's and Morera's theorem shows that integration over $[0,t]$ preserves analyticity, cf. \cite[p. 228]{Koe04}, hence $G$ is analytic on $S_{\theta}$.\\
Obviously, we have $\Phi_{Y_t}(iy) = G(iy)$ for all $y\in \R^d$ by Theorem \ref{cf} and the definition of $G$. Thus, \cite[Lemma A.4]{dfs2003} finally implies $\Phi_{Y_t} \equiv G$  on $S_{\theta}$.
\end{proof}

With Theorem \ref{theorem_analytic_cf} at hand, we can establish the following result:

\begin{theorem}[Absolute integrability]
If (\ref{finite_exp_moment}) holds for some $\epsilon>0$, then $w \mapsto \Phi_{Y_t}(y+iw)$ is absolutely integrable, for all $y\in\R^d$ with $\norm{y}<\theta$, where $\theta$ is given as in Theorem \ref{theorem_analytic_cf}.
\label{cf_abs_int}
\end{theorem}

\begin{proof}
As in the proof of Theorem \ref{theorem_analytic_cf}, we obtain from
\[ \Re(H_{y+iw}(s)) = H_y(s) - \frac12 \int_0^s e^{A^\T s} w w^\T e^{As} \,ds \]
and $\Re(e^{\tr(Z)}) \leq |e^{\tr(Z)}| = e^{\Re(\tr(Z))} = e^{\tr(\Re(Z))}$ for  $Z\in M_d(\C)$, that
\begin{align*} &\Re\( \int_0^t \int_{\S_d^+} \( e^{\tr((H_{y+iw}(s) + \rho^*(y+iw))X)} -1 \) \,\kappa_L(dX) ds \) \\
&\qquad \qquad   \leq  \int_0^t \int_{\S_d^+} \( e^{\tr((H_y(s)+\rho^*(y))X)} - 1 \)  \,\kappa_L(dX) ds. 
\end{align*}
Using this inequality yields
\begin{align*}
&|\Phi_{Y_t}(y+iw)| \\
&\quad \leq \Phi_{Y_t}(y) e^{ - \frac12 \tr(\Sigma_0(e^{A^\T t} \Abf^{-*}(ww^\T) e^{At}- \Abf^{-*}(ww^\T))) - \frac12 \int_0^t \tr(\gamma_L(e^{A^\T s} \Abf^{-*}(ww^\T) e^{As}- \Abf^{-*}(ww^\T))) \,ds } \\
&\quad = \Phi_{Y_t}(y) e^{ - \frac12 \sp{\(\Abf^{-1}\Bscr(t)(\Sigma_0)+\int_0^t\Abf^{-1}\Bscr(s)(\gamma_L)\,ds\) w}{w} }
\end{align*}
with $\Bscr(t)$ as in Lemma \ref{adjointforfourier}. Note that $\Abf^{-1}\Bscr(t)(\Sigma_0)+\int_0^t \Abf^{-1}\Bscr(s)(\gamma_L)\,ds \in \S_d^+$, hence
\[ \int_{\R^d} |\Phi_{Y_t}(y+iw)| \,dw \leq \Phi_{Y_t}(y) \int_{\R^d} e^{ - \frac12 \sp{\(\Abf^{-1}\Bscr(t)(\Sigma_0)+\int_0^t \Abf^{-1}\Bscr(s)(\gamma_L)\,ds\)w}{w} } \,dw < \infty, \]
by Theorem \ref{theorem_analytic_cf}, and because the integrand is proportional to the density of a multivariate Normal distribution.
\end{proof}

\subsection{Martingale Conditions and Equivalent Martingale Measures}\label{sec: emm}

For notational convenience, we work in this section with the model
\beqan
d Y_t &=& (\mu + \beta(\Sigma_t) )\,dt + \Sigma_t^\frac12 \, dW_t + \rho(dL_t), \quad Y_0\in\R^d, \\
d\Sigma_t &=& (\gamma_L + A\Sigma_t + \Sigma_t A^\T) \,dt + dL_t, \quad \Sigma_0\in\S_d^{++},
\eeqan
where $L$ is a driftless matrix subordinator with \levy measure $\kappa_L$. Clearly, this is our multivariate stochastic volatility model of OU type (\ref{sde_Y}), (\ref{sde_Sigma}), except that $\mu$ in (\ref{sde_Y}) is replaced by $\mu-\rho(\gamma_L)$, such that there is no deterministic drift from the leverage term $\rho(dL_t)$.

In mathematical finance, $Y$ is used to model the joint dynamics of the $\log$-returns of $d$ assets  with price processes $S_t^i = S_0^i e^{Y_t^i}$, where we set $Y_0^i=0$ from now on and, hence, $S_0$ denotes the vector of initial prices.

The martingale property of the \emph{discounted stock prices} $(e^{-rt}S_t)_{t\in[0,T]}$ for a constant interest rate $r>0$ can be characterised as follows. 

\begin{theorem}\label{t:2.10}
The discounted price process $(e^{-rt}S_t)_{t\in[0,T]}$ is a martingale if and only if, for $i=1,\ldots,d$,
\begin{equation} \label{e:int}
 \int_{\{\norm{X}>1\}} e^{\rho^i(X)} \,\kappa_L(dX)<\infty,
 \end{equation}
and
\begin{align}
\beta^i(X) &= -\frac12 X_{ii}, \quad  X\in\S_d^+, \label{martcond_beta}\\
\mu_i &= r - \int_{\S_d^+} ( e^{\rho^i(X)} - 1 ) \,\kappa_L(dX) . \label{martcond_wtmu}
\end{align}

\label{th_mart}\end{theorem}

\begin{proof}
Define $\wh{S}_t:=e^{-rt}S_t$ for all $t\in[0,T]$ and let $i \in \{1,\ldots,d\}$. By It\^o's formula and \cite[Proposition III.6.35]{Jacod2003}, $\wh{S}^{\,i}$ is a local martingale if and only if \eqref{e:int}, \eqref{martcond_beta} and \eqref{martcond_wtmu} hold. Thus it remains to show that it is actually a true martingale under the stated assumptions. Since  $\wh{S}$ is a positive local martingale, it is a supermartingale and hence a martingale if and only if $E(\wh{S}^{\,i}_T)=\wh{S}^{\,i}_0$ for all $i \in \{1,\ldots,d\}$. This can be seen as follows. By Theorem \ref{theorem_marginal_dynamics}, (\ref{martcond_beta}) and (\ref{martcond_wtmu}) we have
\begin{align*}
\E(\wh{S}_T^{\,i}) &= \wh{S}^{\,i}_0 \E\left(\exp\left( (\mu^i-r)T  + \beta^i(\Sigma_T^+) + \int_0^T (\Sigma_s^{ii})^\frac12 \,dW_s^i + \rho^i(L_T) \right) \right) \\
&= \wh{S}^{\,i}_0 e^{- T \int_{\S_d^+} (e^{\rho^i(X)}-1) \,\kappa_L(dX)} \E\left( e^{- \frac12 (\Sigma_T^+)^{ii} + \rho^i(L_T)} \E\left( e^{\int_0^T (\Sigma_s^{ii})^\frac12 \,dW_s^i} \bigg| (L_s)_{s\in[0,T]}  \right) \right) \\
&= \wh{S}^{\,i}_0 e^{- T \int_{\S_d^+} (e^{\rho^i(X)}-1) \,\kappa_L(dX)} \E\left( e^{\rho^i(L_T)} \right) \\
&= \wh{S}^{\,i}_0.
\end{align*}
This proves the assertion.
\end{proof}

As in \cite[Theorem 3.1]{Nicolato2003}, it is possible to characterise the set 
of all equivalent martingale measures (henceforth EMMs), if the underlying filtration is generated by $W$ and $L$. More specifically, it follows from the Martingale Representation Theorem (cf.\ \cite[Theorem III.4.34]{Jacod2003}), that the density process $Z_t=E(\frac{dQ}{dP}|\mathscr{F}_t)$ of any equivalent martingale measure $Q$ can be written as  
\beq Z = \Ecal \left( \int_0^\cdot \psi_s dW_s + (Y-1) \ast (\mu^L-\nu^L) \right) \label{eq: density process general}\eeq
for suitable processes $\psi$ and $Y$ in this case. Here $\mu^L$ resp.\ $\nu^L$ denote the random measure of jumps resp. its compensator (cf.\ \cite[Section II.1]{Jacod2003} for more details).
Under an arbitrary  EMM, $L$ may not be a \levy process, and $W$ and $L$ may not be independent. However, there is a subclass of \emph{structure preserving} EMMs under which $L$ remains a \levy process independent of $W$.  This translates into the following specifications of $\psi$ and $Y$ (cf. \cite[Theorem 3.2]{Nicolato2003} for the univariate case):

\begin{theorem}[Structure preserving EMMs]\label{th_sp_emm}
Let $y:\S_d^+\ra(0,\infty)$ such that
\begin{enumerate}
 \item $ \int_{\S_d^+} ( \sqrt{y(X)}-1 )^2 \,\kappa_L(dX) < \infty $,
 \item $ \int_{\{\,\norm{X}>1\}} e^{\rho^i(X)} \, \kappa_L^y(dX) < \infty, \quad i=1,\ldots,d $,
\end{enumerate}
where $\kappa_L^y(B):=\int_{B} y(X)\,\kappa_L(dX)$ for $B\in\Bcal(\S_d^+)$. Define the $\R^d$-valued process $(\psi_t)_{t\in[0,T]}$ as
\[ \psi_t = - \Sigma_t^{-\frac12} \left( \mu + \beta(\Sigma_t) + \frac12 \left(\ba{c} \Sigma_t^{11} \\ \vdots \\ \Sigma_t^{dd} \ea\right) + \left(\ba{c} \int_{\S_d^+}(e^{\rho^1(X)}-1)\,\kappa_L^y(dX) \\ \vdots \\ \int_{\S_d^+}(e^{\rho^d(X)}-1)\,\kappa_L^y(dX) \ea\right) - \1 r \right),  \]
where $\1=(1,\ldots,1)^\T\in\R^d$. Then $Z=\Ecal( \int_0^\cdot \psi_s dW_s + (y-1) \ast (\mu^L-\nu^L) )$ is a density process, and the probability measure $Q$ defined by $\frac{dQ}{dP}=Z_T$ is an equivalent martingale measure. Moreover, $W^Q:=W-\int_0^\cdot \psi_s ds$ is a $Q$-standard Brownian motion, and $L$ is an independent driftless $Q$-matrix subordinator with \levy measure $\kappa_L^y$. The $Q$-dynamics of $(Y,\Sigma)$ are given by
\beqa
dY_t^i &=& \(r - \int_{\S_d^+}(e^{\rho^i(X)}-1)\,\kappa_L^y(dX) - \frac12 \Sigma_t^{ii} \) \,dt + \(\Sigma_t^\frac12 \,dW_t^Q\)^i + \rho^i(dL_t), \quad i=1,\ldots,d, \\
d\Sigma_t &=& (\gamma_L + A\Sigma_t + \Sigma_t A^\T) \,dt + dL_t.
\eeqa
\end{theorem}

\begin{proof} Since $y-1>-1$, $Z$ is strictly positive by \cite[Theorem I.4.61]{Jacod2003}. The martingale property of $Z$ follows along the lines of the proof of \cite[Theorem 3.2]{Nicolato2003}. The remaining assertions follow from \cite[Proposition 1]{Kallsen2004} and the L\'evy-Khintchine formula by applying the Girsanov-Jacod-M\'emin Theorem as in \cite[Proposition 4] {Kallsen2004} to the $\R^{\frac12 d(d+1)}$-valued process
\[ \wt{L} = \(\ba{c} W^Q \\ 0 \ea\)+\vech(L), \]
where  $W^Q:=W-\int_0^\cdot \psi_s ds$. 
\end{proof}

The previous theorem shows that it is possible to use a model of the same type under the real-world probability measure $P$ and some EMM $Q$, e.g., to do option pricing and risk management within the same model class. The model parameters under $Q$ can be determined by calibration, the model parameters under $P$ by statistical methods.

\section{Option pricing using integral transform methods}\label{pricing}

In this section we first recall results of \cite{Eberlein2008} on Fourier pricing in general multivariate semimartingale models. To this end, let $S=(S_0^1 e^{Y^1},\ldots,S_0^d e^{Y^d})$ be a $d$-dimensional semimartingale such that the discounted price process $(e^{-rt}S_t)_{t\in[0,T]}$ is a martingale under some pricing measure $Q$, for some constant instantaneous interest rate $r>0$.

We want to determine the price $E_Q(e^{-rT}f(Y_T-s))$ of a European option with payoff $f(Y_T-s)$ at maturity $T$, where $f:\R^d\ra\R_+$ is a measurable function and $s:=(-\log(S_0^1),\ldots,-\log(S_0^d))$. Denote by $\wh{f}$ the \emph{Fourier transform} of $f$. The following theorem is from \cite[Theorem 3.2]{Eberlein2008} and represents a multivariate generalisation of integral transform methods first introduced in the context of option pricing by \cite{Carr1999} and \cite{raible.2000}.

\begin{theorem}[Fourier Pricing]\label{derivative_fourier_inversion}
Fix $R\in\R^d$, let $g(x):=e^{-\sp{R}{x}} f(x)$ for $x\in\R^d$, and assume that
\begin{enumerate}
\item[($i$)] $g\in L^1\cap L^\infty$, \quad ($ii$)\; $\Phi_{Y_T}(R) < \infty$, \quad ($iii$)\; $w \mapsto \Phi_{Y_T}(R+iw)$ belongs to $L^1$.
\end{enumerate}
Then, 
\beq E_Q(e^{-rT}f(Y_T-s)) = \frac{e^{-\sp{R}{s}-rT} }{(2\pi)^d} \int_{\R^d} e^{-i\sp{u}{s}}  \Phi_{Y_T}(R+iu) \wh{f}(iR-u) \,du \label{derivative_fourier_inversion_eq}.
\eeq
\end{theorem}

Observe that Theorems \ref{theorem_analytic_cf} and \ref{cf_abs_int} show that Conditions \textit{(ii)} and \textit{(iii)} are satisfied for our multivariate stochastic volatility model of OU type (\ref{sde_Y}), (\ref{sde_Sigma}) if condition (\ref{finite_exp_moment}) holds, i.e., if $L$ has enough exponential moments. More specifically, the vector $R$ has to lie in the intersection of the domains of $\Phi_{Y_T}$ and $\wh{f}$.

We now present some examples. As is well-known, the Fourier transform of the payoff function of a \emph{plain vanilla call option} with strike $K>0$, $f(x)=(e^x-K)^+$ is given by
\beq \wh{f}(z) = \frac{K^{1+iz}}{iz(1+iz)} \label{ft_vanilla_call}\eeq
for $z\in\C$ with $\Im(z)>1$. The Fourier transforms of many other single-asset options like barrier, self-quanto and power options as well as multi-asset options like worst-of and best-of options can be found, e.g., in the survey  \cite{Eberlein2008}. From the unpublished paper of \cite{Hubalek} we have the following formulae for basket and spread options.

\begin{example}
\begin{enumerate}
 \item The Fourier transform of $f(x)=(K-\sum_{j=1}^d e^{x_j})^+$, $K>0$, that is the payoff function of a \emph{basket} put option, is given by
\[\wh{f}(z) = K^{1 + i \sum_{j=1}^d z_j} \frac{ \prod_{j=1}^d \Gamma(iz_j) }{\Gamma(2+i \sum_{j=1}^d z_j)} \]
for all $z\in\C^d$ with $\Im(z_j)<0$, $j=1,\ldots,d$. The price of the corresponding call can easily be derived using the put-call-parity $(K-x)^+=(x-K)^+-x+K$. Since we have separated the initial values $s$ in (\ref{derivative_fourier_inversion_eq}), we can use FFT methods to compute the prices of weighted baskets for several weights efficiently.
 \item The Fourier transform of the payoff function of a \emph{spread} call option, $f(x)=(e^{x_1}-e^{x_2}-K)^+$, $K>0$, is given by
\[\wh{f}(z) = \frac{K^{1 + iz_1 + iz_2}}{iz_1(1+iz_1)}  \frac{\Gamma(i z_2) \Gamma(-iz_1-iz_2-1)}{\Gamma(-iz_1-1)} \]
for all $z\in\C^2$ with $\Im(z_1)>1$, $\Im(z_2)<0$ and $\Im(z_1+z_2)>1$, see also \cite{hurd2010}.
\end{enumerate}
\end{example}

Since the Fourier transform of $(e^{x_1}-e^{x_2})^+$ does not exist anywhere, we cannot use Theorem \ref{derivative_fourier_inversion} to price zero-strike spread options. Nevertheless, we can derive a similar formula directly. Alternatively, one could use the change of numeraire technique of \cite{margrabe78}, which would lead to formulae of a similar complexity.

\begin{proposition}[Spread options with zero strike]\label{prop_spread_no_strike}
Suppose that
\[ \Phi_{(Y_T^1,Y_T^2)}(R,1-R)<\infty \quad\textit{ for some } R>1. \]
Then the price of a \emph{zero-strike spread option} with payoff $(S_0^1 e^{Y_T^1}-S_0^2 e^{Y_T^2})^+$ is given by
\[
\E_Q(e^{-rT}(S_T^1-S_T^2)^+) = \frac{e^{R(s_2-s_1)-s_2-rT}}{2\pi} \int_{\R} e^{iu(s_2-s_1)} \frac{\Phi_{(Y_T^1,Y_T^2)}(R+iu,1-R-iu)}{(R+iu)(R+iu-1)} \, du,
\]
where $s_1=-\ln(S_0^1)$ and $s_2=-\ln(S_0^2)$.
\end{proposition}

Observe that unlike for $K>0$, one only has to compute a one-dimensional integral to determine the price of a zero-strike spread option. This will be exploited in the calibration procedure in Section 4.

\begin{proof}
Let $R>1$ and define $f_K(x)=(e^x-K)^+$ for $K>0$, and $g_K(x)=e^{-Rx}f_K(x)$. By Fourier inversion and (\ref{ft_vanilla_call}), we have
\[ f_{e^y}(x) = \frac{1}{2\pi} \int_\R \frac{ e^{(R+iu)x} e^{(1-R-iu)y} }{ (R+iu)(R+iu-1) }  \,du, \]
for all $y\in\R$. Hence, for the function $h_{e^y}(x) := (S_0^1 e^x-S_0^2 e^y)^+ = f_{e^{y-s_2}}(x-s_1)$ we get
\[ h_{e^y}(x) = \frac{1}{2\pi} e^{R(s_2-s_1)-s_2} \int_\R e^{iu(s_2-s_1)} \frac{ e^{(R+iu)x} e^{(1-R-iu)y} }{ (R+iu)(R+iu-1) }  \,du .\]
Finally, by Fubini's theorem
\begin{align*}
\E_Q(h_{e^{Y_T^2}}(Y_T^1)) &= \frac{e^{R(s_2-s_1)-s_2}}{2\pi} \int_{\R^2} \int_\R e^{iu(s_2-s_1)} \frac{ e^{(R+iu)x} e^{(1-R-iu)y} }{ (R+iu)(R+iu-1) }  \,du \, P_{(Y_T^1,Y_T^2)}(dx,dy) \\
&= \frac{e^{R(s_2-s_1)-s_2}}{2\pi} \int_{\R} e^{iu(s_2-s_1)} \frac{\Phi_{(Y_T^1,Y_T^2)}(R+iu,1-R-iu)}{(R+iu)(R+iu-1)} \, du ,
\end{align*}
where the application of Fubini's theorem is justified by
\begin{align*}
\int_{\R^2} \int_\R \Big| \frac{ e^{(R+iu)x} e^{(1-R-iu)y} }{ (R+iu)(R+iu-1) } \Big|  \,du \, P_{(Y_T^1,Y_T^2)}(dx,dy) &= \int_{\R^2} e^{Rx} e^{(1-R)y} \int_\R | \wh{g_1}(u) | \,du \, P_{(Y_T^1,Y_T^2)}(dx,dy) \\
&\leq \norm{\wh{g_1}}_{\L^1} \Phi_{(Y_T^1,Y_T^2)}(R,1-R) < \infty,
\end{align*}
since $\norm{\wh{g_1}}_{\L^1}<\infty$ as shown in \cite[Example 5.1]{Eberlein2008}.
\end{proof}

\section{Calibration of the OU-Wishart model}\label{section: OU-Wishart}\label{sec: OU-Wishart}

We now put forward a specific parametric specification of the model discussed in Section \ref{s:model}. To this end, let $n\in\N$, $\Theta\in\S_d^+$ and let $X$ be a $d\times n$ random matrix with i.i.d.\ standard normal entries. Then, the matrix $M:=\Theta^\frac12 XX^\T \Theta^\frac12$ is said to be \emph{Wishart distributed}, written $M \thicksim\W_d(n,\Theta)$. Note that this definition can be extended to noninteger $n>d-1$ using the characteristic function
\beq Z \mapsto \det(I_d-2iZ\Theta)^{-\frac12 n}, \label{eq: cf Wishart}\eeq
see \cite[Theorem 3.3.7]{grupta}. Since $M\in\S_d^+$ almost surely, we can define a compound Poisson matrix subordinator $L$ with intensity $\lambda$ and $\W_d(n,\Theta)$ distributed jumps. We call the resulting multivariate stochastic volatility model of OU type  \emph{OU-Wishart model}.

\begin{remark}
There exists a subclass of structure preserving EMMs $Q$ (cf. Theorem \ref{th_sp_emm}) such that we have an OU-Wishart model under both $P$ and $Q$. This means that $L$ is a compound Poisson process with $\W_d(n,\Theta)$ distributed jumps and intensity $\lambda$ under $P$, and $\W_d(\wt{n},\wt{\Theta})$ distributed jumps with intensity $\wt{\lambda}$ under $Q$. We only need to assume that the Wishart distribution under both $P$ and $Q$ has a Lebesgue density, i.e., $n,\wt{n}>d-1$ and $\Theta,\wt{\Theta}\in\S_d^{++}$. Then, one simply has to take $y$ as the quotient of the respective \levy densities. Hence, by \cite[3.2.1]{grupta}, $y$ has to be defined as
\[ y(X) = \frac{\wt{\lambda}}{\lambda} \(2^{\frac12 (\wt{n}-n) d} \frac{\Gamma_d\(\frac12 \wt{n}\)}{\Gamma_d\(\frac12 n\)} \frac{\det(\wt{\Theta})^{\frac12 \wt{n}}}{\det(\Theta)^{\frac12 n}} \)^{-1} \det(X)^{\frac12(\wt{n}-n)} e^{-\frac12\tr\((\wt{\Theta}^{-1}-\Theta^{-1}) X\)}, \quad X\in\S_d^+. \]
\end{remark}

Since we have $\int_{\S_d^+} e^{\tr(RX)} \, \kappa_L(dX) = \lambda \det(I_d-2R\Theta)^{-\frac12 n}$ by (\ref{eq: cf Wishart}), we see that the compound Poisson process $L$ has exponential moments as long as $\norm{R}<\frac{1}{2\norm{\Theta}}$, where $\norm{\cdot}$ denotes the spectral norm. Consequently, (\ref{finite_exp_moment}) holds for $\epsilon:=\frac{1}{2\norm{\Theta}}$, and we can apply the integral transform methods from the previous section to compute prices of multi-asset options.

Note that for the particularly simple special case of diagonal $A$, $\beta$ and $\rho$, each asset follows a BNS model at the margins by \eqref{sde: Y marginal} and \eqref{sde: Sigma marginal}. In particular, for $n=2$ we see that $L^{ii}$, $i=1,\ldots,d$, is a compound Poisson subordinator with exponentially distributed jumps, thus we have in distribution the \emph{$\Gamma$-OU BNS model} with stationary Gamma distribution at the margins, cf., e.g., \cite[Section 2.2]{Nicolato2003}. Then, the characteristic functions of the single assets are known in closed form. Note that while the characteristic function of the stationary distribution of the marginal OU type process is still known for $n \neq 2$, it no longer corresponds to a Gamma distribution in this case.

\subsection{The OU-Wishart model in dimension 2}\label{sec: 2dim OU-Wishart}

We work directly under a pricing measure $Q$ and consider the following specific two-dimensional case of our model, where we restrict ourselves in particular to a diagonal mean-reversion matrix $A$ and a leverage term  $\rho$ such that both jumps of the respective variance and of the covariance enter the price. Our model is given by
\begin{align*}
\bpm dY_t^1 \\ dY_t^2 \epm &= \left( \bpm \mu_1 \\ \mu_2 \epm - \frac12 \bpm \Sigma_t^{11} \\ \Sigma_t^{22} \epm \right) \,dt + \bpm \Sigma_t^{11} & \Sigma_t^{12} \\ \Sigma_t^{12} & \Sigma_t^{22} \epm^\frac12 \bpm dW_t^1 \\ dW_t^2 \epm + \bpm \rho_1 \,dL_t^{11} + \rho_{12}\,dL_t^{12} \\ \rho_2 \,dL_t^{22} + \rho_{21}\,dL_t^{12} \epm \\
\bpm d\Sigma_t^{11} & d\Sigma_t^{12} \\ d\Sigma_t^{12} & d\Sigma_t^{22} \epm &= \left( \bpm \gamma_1 & 0 \\ 0 & \gamma_2 \epm + \bpm 2 a_1 \Sigma_t^{11} & (a_1+a_2) \Sigma_t^{12} \\ (a_1+a_2) \Sigma_t^{12} & 2 a_2 \Sigma_t^{22} \epm \right)\,dt +  \bpm dL_t^{11} & dL_t^{12} \\ dL_t^{12} & dL_t^{22} \epm
\end{align*}
with initial values
\[ Y_0 = \bpm 0 \\ 0 \epm, \quad \Sigma_0 = \bpm \Sigma_0^{11} & \Sigma_0^{12} \\ \Sigma_0^{12} & \Sigma_0^{22} \epm \in \S_2^{++} , \]
and parameters $\gamma_1,\gamma_2 \geq 0,\, a_1,a_2<0,\, \rho_1,\rho_2,\rho_{12},\rho_{21}\in\R$. $L$ is a compound Poisson process with intensity $\lambda$ and $\W_2(n,\Theta)$-jumps, where $n=2$ and
\[ \Theta = \bpm \Theta_{11} & \Theta_{12} \\ \Theta_{12} & \Theta_{22} \epm\in\S_2^+.\]
Therefore, all components of $L$ jump at the same time. Since the second order properties of the Wishart distribution are known explicitly, cf. \cite[Theorem 3.3.15]{grupta}, the covariances of the jumps are given by
\begin{align*}
\mathrm{Cov}(\Delta L_t^{11},\Delta L_t^{12}|\Delta L_t^{11}\neq 0) &= 4\Theta_{11}\Theta_{12}, \\
\mathrm{Cov}(\Delta L_t^{22},\Delta L_t^{12}|\Delta L_t^{11}\neq 0) &= 4\Theta_{22}\Theta_{12} , \\
\mathrm{Cov}(\Delta L_t^{11},\Delta L_t^{22}|\Delta L_t^{11}\neq 0) &= 4\Theta_{12}^2.
\end{align*}
This shows that even if $\rho$ is diagonal, i.e., $\rho_{12}=0=\rho_{21}$, the leverage terms of both assets are correlated. If $\rho$ is non-diagonal, then $\theta_{12}$ also influences the marginal distribution of each asset.

\paragraph{Multi-asset option pricing}

By (\ref{eq: Phi}) and (\ref{eq: cf Wishart}), the joint moment generating function of $(Y^1,Y^2)$ is given by
\begin{align*}
&\E[e^{y^\T Y_t}] =\\
& \exp \left( y^\T \mu t  + \tr(\Sigma_0 H_y(t)) +  \int_0^t \tr(\gamma_L H_y(s)) ds + \lambda  \int_0^t \frac{1}{\det(I_2-2(H_y(s)+\rho^*(y))\Theta)} \,ds - \lambda t \right) 
\end{align*}
with $H_y$ as in (\ref{eq: H}), $A=\(\bsm a_1 & 0 \\ 0 & a_2 \esm\)$, $\gamma_L=\(\bsm \gamma_1 & 0 \\ 0 & \gamma_2 \esm\)$, and $\rho^*(y)= \( \bsm \rho_1 y_1 & \rho_{12}y_1 \\ \rho_{21}y_2 & \rho_2 y_2 \esm \)$. It does not seem to be possible to obtain a closed form expression in terms of ordinary functions, unless one sets $a_1=a_2=:a$. In this case, if $\Delta = \sqrt{4b_0b_2-b_1^2}\neq 0$, one has
\begin{align*}
\E[e^{y_1 Y_t^1 + y_2 Y_t^2}] =& \exp\left\{ y_1 \mu_1 t + y_2 \mu_2 t + \frac{e^{2at}-1}{4a} \tr\left(\Sigma_0 \bpm  y_1^2-y_1 & y_1 y_2 \\ y_1 y_2 & y_2^2-y_2 \epm \right)\right. \\ 
& \left. \qquad + \frac{1}{4a}\(\gamma_1(y_1^2-y_1)+\gamma_2(y_2^2-y_2)\)\(\frac{1}{2a}(e^{2at}-1)-t\) \right. \\
& \left. \qquad + \frac{\lambda}{2a b_0}\left[ \frac{b_1}{\Delta} \left( \arctan\left(\frac{2b_2+b_1}{\Delta}\right) - \arctan\left(\frac{2b_2e^{2at}+b_1}{\Delta}\right) \right) \right.\right. \\
& \left.\left. \qquad + \frac12 \ln\left(\frac{b_0+b_1+b_2}{b_2e^{4at}+b_1e^{2at}+b_0}\right) \right] +  \frac{\lambda}{b_0} t - \lambda t \right\}
\end{align*}
with coefficients
\begin{align*}
b_0 &:= 1+4\det(B-C)+2\tr(B-C), \\
b_1 &:= -8\det(B)+4\tr(B)\tr(C)-4\tr(BC)-2\tr(B), \\
b_2 &:= 4\det(B), \\
\Delta &:= \sqrt{4b_0b_2-b_1^2},
\end{align*}
and matrices
\begin{align*}
B := \frac{1}{4a} \bpm  y_1^2-y_1 & y_1 y_2 \\ y_1 y_2 & y_2^2-y_2 \epm \Theta, \quad C := \bpm \rho_1 y_1 & \rho_{12}y_1 \\ \rho_{21} y_2 & \rho_2 y_2 \epm \Theta.
\end{align*}
Note that $\arctan$ has to be understood as a function of complex argument to cover the case where the term in the square root of $\Delta$ is negative. If $\Delta=0$, we obtain
\begin{align*}
&\E[e^{y_1 Y_t^1 + y_2 Y_t^2}] =\\
& \quad \exp\left\{ y_1 \mu_1 t + y_2 \mu_2 t + \frac{e^{2at}-1}{4a} \tr\left(\Sigma_0 \bpm  y_1^2-y_1 & y_1 y_2 \\ y_1 y_2 & y_2^2-y_2 \epm \right)\right. \\ 
& \quad \left. \qquad + \frac{1}{4a}\(\gamma_1(y_1^2-y_1)+\gamma_2(y_2^2-y_2)\)\(\frac{1}{2a}(e^{2at}-1)-t\) \right. \\
& \quad \left. \qquad + \frac{\lambda}{2a b_0}\left[ \frac{b_1}{2b_2e^{2at}+b_1} - \frac{b_1}{2b_2+b_1} + \frac12 \ln\left(\frac{b_0+b_1+b_2}{b_2e^{4at}+b_1e^{2at}+b_0}\right) \right] +  \frac{\lambda}{b_0} t - \lambda t \right\}.
\end{align*}
Using $\det(A+B)=\det(A)+\det(B)+\tr(A)\tr(B)-\tr(AB)$ for $A,B\in M_2(\R)$, the above formulae follow from 
\[ \det(I_2-2(H_y(s)+\rho^*(y))\Theta) = \det(I_2-2(e^{2as}-1)B-2C) = b_0 + b_1 e^{2as} + b_2 e^{4as}, \]
and straightforward integration. Likewise, one can also derive a closed form expression for $n=4,6,\ldots$ using \cite[2.18(4)]{TISP}. 

Consequently, one faces a tradeoff at this point. One possibility is to retain the flexibility of different mean reversion speeds $a_i$ by evaluating the remaining integral using numerical integration. Alternatively, one can restrict attention to identical mean reversion speeds in order to have a closed-form expression of the moment generating function at hand. The impact of this decision on the calibration performance is discussed in Section \ref{ss:empirical} below.

\paragraph{Single-asset option pricing}
For pricing single-asset options, one only needs the transforms of the marginal models, such that the above expressions simplify considerably. For example, the moment generating function of $Y^1$ is given by
\begin{align*}
\E[e^{y_1 Y_t^1}] = \exp\bigg\{ y_1\mu^1 t &+ \frac{e^{2a_1t}-1}{4a_1} (y_1^2-y_1) \Sigma_0^{11}  + \frac{1}{4a_1}\left(\frac{1}{2a_1}(e^{2a_1 t}-1)-t\right) (y_1^2-y_1) \gamma_L^{11}  \nonumber\\
 &+ \frac{\lambda}{2a_1 b_0} \ln\(\frac{b_0+b_1}{b_0+e^{2a_1 t}b_1}\) + \frac{\lambda t}{b_0} - \lambda t \bigg\}, 
\end{align*}
where $b_0$ and $b_1$ simplify to
\begin{align*}
b_0 &= 1 + \(\frac{1}{2a_1} (y_1^2-y_1)-2\rho_1 y_1\)\Theta_{11} - 2\rho_{12} y_1 \Theta_{12}, \\
b_1 &= -\frac{1}{2a_1} (y_1^2-y_1) \Theta_{11}.
\end{align*}
Note that one can use the recursion formula stated in \cite[2.155]{TISP} to obtain a closed form expression for $\W_2(n,\Theta)$-jumps with $n \in 2 \mathbb{N}$, too. In the special case where the operator $\rho$ is diagonal, i.e., if $\rho_{12}=\rho_{21}=0$, the margins are (in distribution) $\Gamma$-OU BNS models, whose moment generating function has been derived in \cite[Table 2.1]{Nicolato2003}.

\begin{remark}[High Dimensionality]
The above model can also be defined for $d>2$, but of course, the Fourier formula \eqref{derivative_fourier_inversion_eq} becomes numerically infeasible in high dimensions. Nevertheless, if $\rho$ is diagonal, the calibration of a high dimensional OU-Wishart model is still possible by only
evaluating options on \emph{two} underlyings. Using zero strike spread
options and provided the characteristic function is known explicitly, this means
that one only has to evaluate single integrals numerically, as in the
univariate case. Indeed, combining \cite[Proposition 4.5]{Barndorff-Nielsen2009} and the fact that every symmetric sub-matrix of a Wishart distributed matrix is again Wishart distributed, cf. \cite[Theorem 3.3.10]{grupta}, it follows that the joint dynamics of each pair of assets follows a $2$-dimensional OU-Wishart model as above. Hence, we can calibrate the model using only two-asset options (e.g., spread options).  The price to pay is that the resulting model only incorporates pairwise dependencies, since the respective covariances completely determine the underlying Wishart distribution.
\end{remark}

\begin{remark}
If $\rho$ is diagonal, we have equivalence in distribution of the margins of our model to a $\Gamma$-BNS model. This implies immediately that we need to use prices on multi-asset options in order to infer all parameters from observed option prices. If $\rho$ is non-diagonal, we have a $\Gamma$-BNS model with an additional (correlated) jump term. Due to this additional term, it might be possible to infer $\Theta_{12}$ from single-asset options. However, one cannot obtain $\Sigma_0^{12}$ in this way because it does not appear in the marginal moment generating function.

In many multi-factor univariate models one can in general similarly not be sure whether one can uniquely determine all parameters from observed option prices. In many papers the parameters are calibrated and the procedure seems to work, but we are not aware of any reasonably complex multi-factor model where the identifiability of the parameters based on option prices has been established. The reason is clearly the highly nontrivial relation between the parameters and the option prices.
\end{remark}

\subsection{Empirical illustration}\label{ss:empirical}

The aim of this subsection is to show that a calibration of the OU-Wishart model to market prices is feasible. Since multi-asset options are mostly traded over-the-counter, it is difficult to obtain real price quotes. To circumvent this problem, we proceed as in \cite{Taylor2009} and consider \emph{foreign exchange rates} instead, where a call option on some exchange rate can be seen as a spread option between two others. Let us emphasise that our calibration routine should not be seen as a finished product, but much rather as a first test and proof of principle. A more detailed investigation as well as an extension to numerically more involved models with non-diagonal $A$ is left to future research.

We consider a $2$-dimensional OU-Wishart model as above. Our first asset is the EUR/USD exchange rate $S^{\$/\eur}=S_0^{\$/\eur}e^{Y^1}$, that is, the price of 1 \eur\ in \$, and our second asset is the GBP/USD exchange rate $S^{\$/\pounds}=S_0^{\$/\pounds}e^{Y^2}$, i.e., the price of 1 \pounds\ in \$. We model directly under a martingale measure. Therefore we have, by Theorem \ref{th_mart}, that
\begin{align*}
\mu_1 = r_{\$}-r_{\eur} - \int_{\S_d^+} ( e^{\rho_1 X^{11} + \rho_{12} X^{12}} - 1 ) \,\kappa_L(dX).
\end{align*}
Since $\kappa_L$ is the intensity $\lambda$ times a Wishart distribution with parameters $n=2$ and $\theta$, this simplifies to
\begin{align*}
\mu_1 =& r_{\$}-r_{\eur} - \lambda \left(\det\left(I_2-2(\bsm \rho_1 & \rho_{12} \\ 0 & 0 \esm) \Theta\right)^{-1} -1\right) \\
=& r_{\$}-r_{\eur} - \lambda \frac{2\rho_1\Theta_{11}+2\rho_{12}\Theta_{12}}{1-2\rho_1\Theta_{11}-2\rho_{12}\Theta_{12}}.
\end{align*}
Likewise we have
\begin{align*}
\mu_2 = r_{\$}-r_{\pounds} - \lambda \frac{2\rho_2\Theta_{22}+2\rho_{21}\Theta_{12}}{1-2\rho_2\Theta_{22}-2\rho_{21}\Theta_{12}}.
\end{align*}
Thus, for $\rho_{12}=0$ or $\rho_{21}=0$, we recover the martingale conditions of the $\Gamma$-OU BNS model.
By \cite[13.4]{Hull2003}, it follows that the price in \$\ of a plain vanilla call option on $S^{\$/\eur}$ or $S^{\$/\pounds}$ is given by $e^{-r_{\$}T}\E((S_T^{\$/\eur}-K)^+)$ or $e^{-r_{\$}T}\E((S_T^{\$/\pounds}-K)^+)$, respectively. Now observe that the  \$-payoff of a call option on the EUR/GBP exchange rate $S^{\pounds/\eur}$ is given by $S_T^{\$/\pounds}(S_T^{\pounds/\eur}-K)^+ = (S_T^{\$/\eur}-K S_T^{\$/\pounds})^+$, hence it can be regarded as a spread option on $S^{\$/\eur}-S^{\$/\pounds}$ where the initial value of the second asset is replaced by $K S_0^{\$/\pounds}$. Since it is a zero-strike spread option, we can use Proposition \ref{prop_spread_no_strike} to valuate it.

\begin{figure}[hp] 
\includegraphics[width=\linewidth]{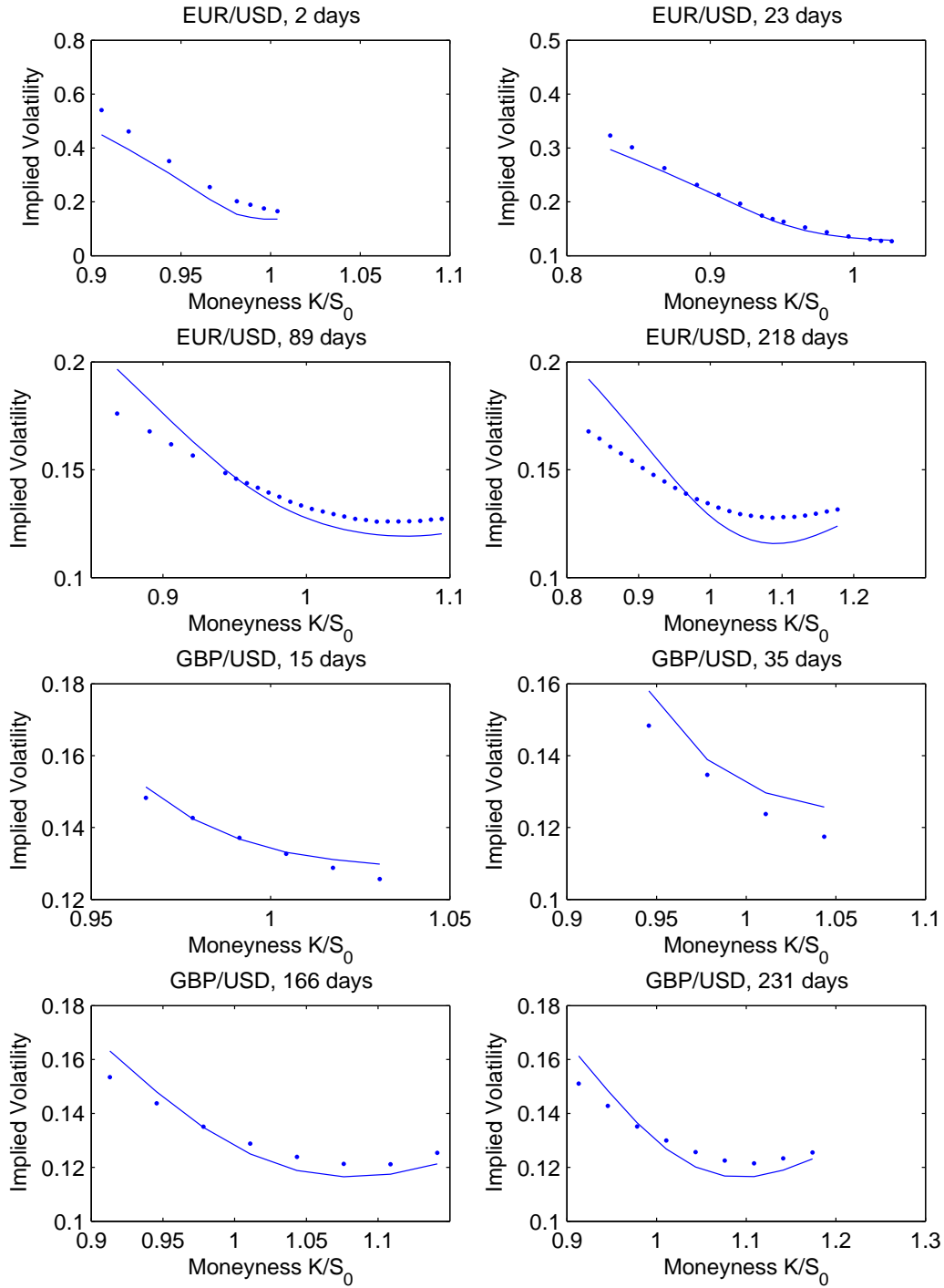}
\caption[]{Comparison of the Black-Scholes implied volatility of market prices (dots) and model prices (solid line). The plots only show the results for the 12-parameter OU Wishart model (Step A), since they do not change visually for the more complex models from Step B to D.}
\label{fig:impvol1}
\end{figure}

\begin{figure}[h] 
\includegraphics[width=\linewidth]{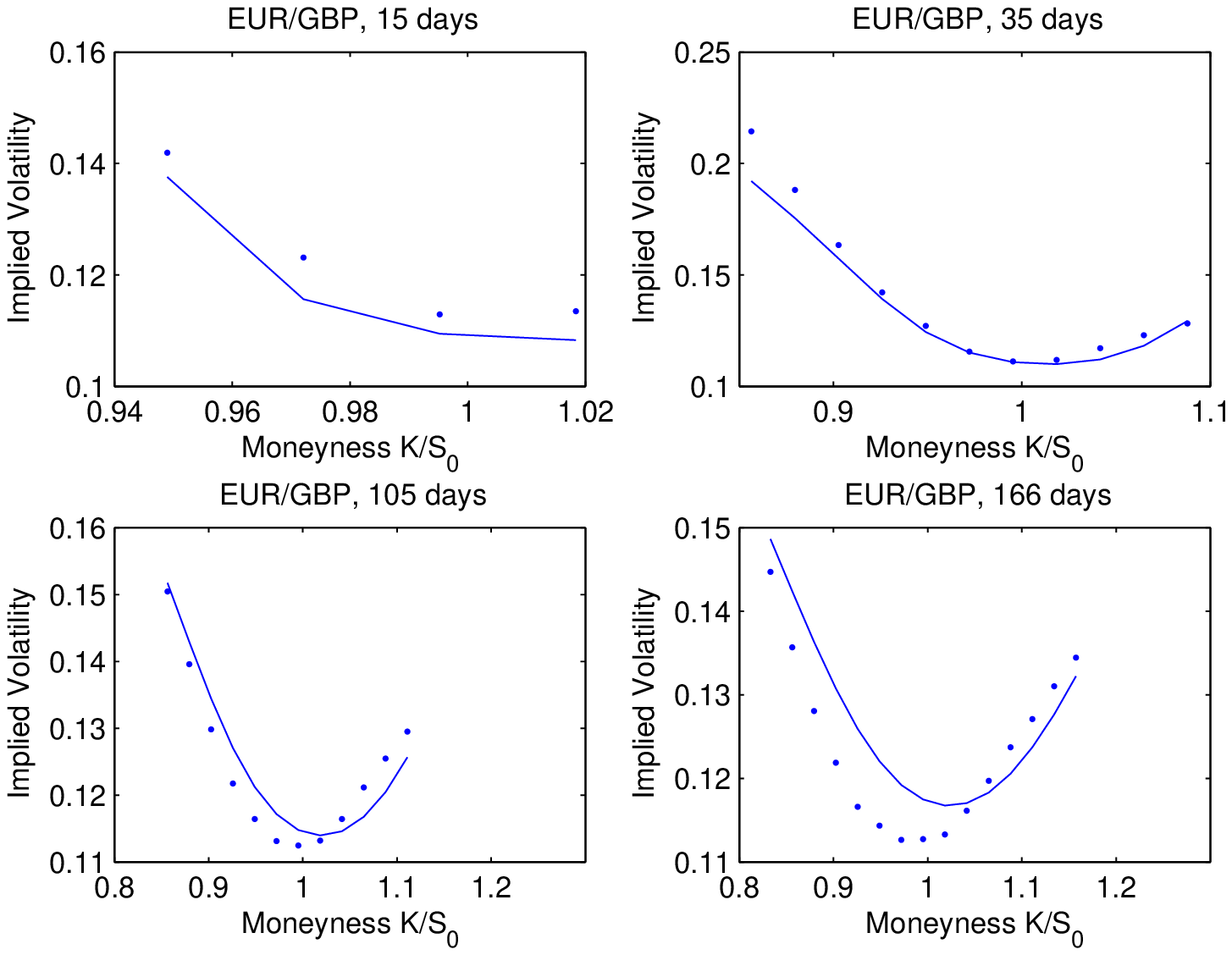}
\caption[]{Comparison of the Black-Scholes implied volatility of market prices (dots) and model prices (solid line). The plots only show the results for the 12-parameter OU Wishart model (Step A), since they do not change visually for the more complex models from Step B to D.}
\label{fig:impvol2}
\end{figure}

We obtained the option price data from EUWAX on April 29, 2010, at the end of the business day. The EUR/USD exchange rate at that time was $S_0^{\$/\eur}=1.3249\$$, the GBP/USD exchange rate was $S_0^{\$/\pounds}=1.5333 \$ $ and the EUR/GBP exchange rate was $0.8641\pounds$. As a proxy for the instantaneous riskless interest rate we took the 3-month LIBOR for each currency, viz. $r_{\eur}=0.604\%$, $r_{\pounds}=0.344\%$ and $r_{\$}=0.676\%$. All call options here are plain vanilla call options of European style. We used 148 call options on the EUR/USD exchange rate, 67 call options on the GBP/USD exchange rate, and 105 call options on the EUR/GBP exchange rate, all of them for different strikes and different maturities, for a total of 320 option prices. We always used the mid-value between bid and ask price. A spread sheet containing all data used for the calibration can be found on the second author's website.

The calibration was performed by choosing the model parameters so as to minimise the root mean squared error (RMSE) between the Black-Scholes volatilities implied by market resp.\ model prices. Note that the RMSE is the square root of the sum of the squared distances divided by the number of options. All computations were carried out in MATLAB and performed on a standard desktop PC with a $2.4$GHz processor.

In \emph{Step A}, we impose $a:=a_1=a_2$ and $\rho_{12}=0=\rho_{21}$, i.e., we make the assumption that the mean reversion parameters of both assets are equal, and that $\rho$ is diagonal. This is the most tractable case, since there is a closed form expression for the moment generating function of $(Y^1,Y^2)$ and the number of model parameters is reduced to 12. The starting and calibrated parameters can be found in Table \ref{table: par}. The overall RMSE is 0.0082, and the run time was 48 minutes, i.e., calibration of the model is feasible even on a standard PC. If one considers only the marginal models for EUR/USD and GBP/USD one has a RMSE of 0.0106 and 0.0048 respectively.  For visualisation, we provide Figure \ref{fig:impvol1} and \ref{fig:impvol2}, where market and model prices are compared in terms of Black-Scholes implied volatility for a few selected maturities. These results illustrate that even this simple model is able to fit the observed smiles rather well.  For comparison, we calibrated two independent univariate $\Gamma$-OU BNS models to the margins separately (see Table  \ref{table: par}) and obtained a lower RMSE of 0.0071 and 0.0020 respectively. This stems from the fact that the additional dependence parameters do not enter the pricing formulas for single asset options, whereas the intensity of the compound Poisson process is the same for all assets in our multivariate framework, unlike when using two univariate models. This means that we are \emph{not overfitting} the marginal distributions with an excessive amount of additional parameters, but much rather using a simplified version of a standard model. Nevertheless, the calibration still performs quite well even when using this simplification. 

\begin{figure}[h] 
\includegraphics[width=\linewidth]{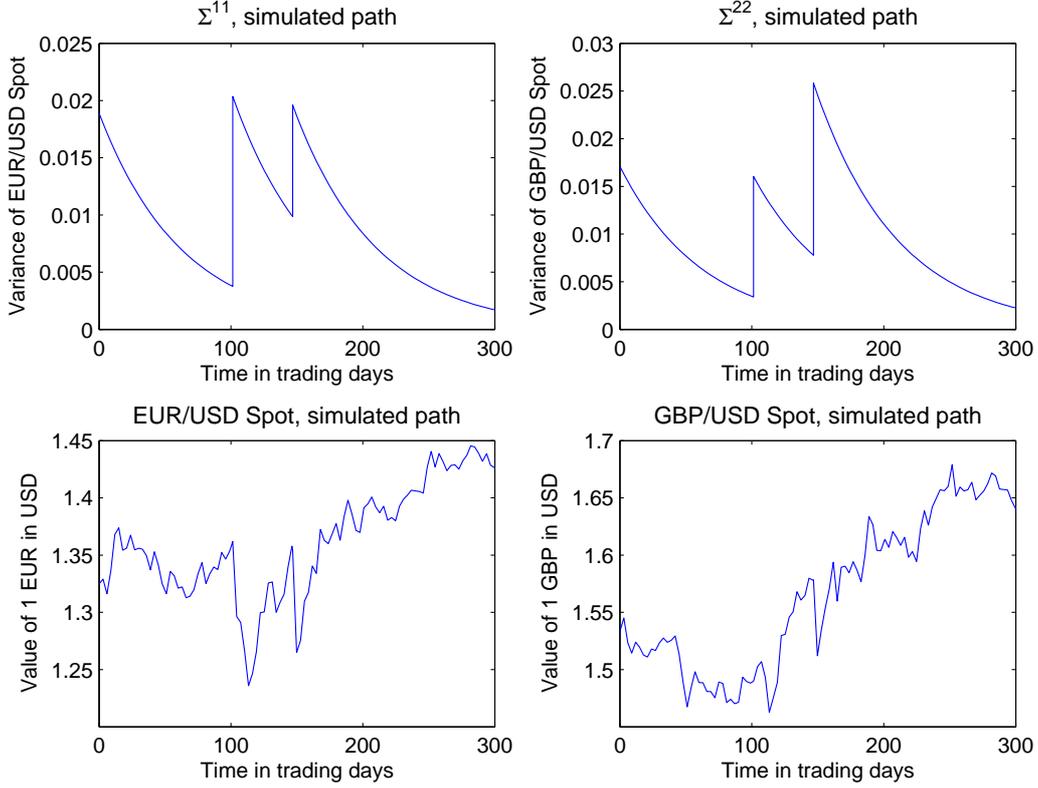}
\caption[]{Simulated sample paths of the EUR/USD and the GBP/USD spot rates and their variances.}
\label{fig:simulation}
\end{figure}

As a further cross-check, Figure \ref{fig:simulation} depicts sample paths of the EUR/USD and the GBP/USD spot rates and their variances, simulated with our calibrated parameters, which show reasonable path properties.

In \emph{Step B}, we allow for a non-diagonal leverage operator $\rho$. Although this introduces two additional parameters, $\rho_{12}$ and $\rho_{21}$, a closed form expression for the moment generating function is still available. As initial values, we take the parameters obtained in Step A and set $\rho_{12}$ and $\rho_{21}$ to zero. After 80 minutes, the optimizer finds a minimum with a RMSE of $0.0079$. At the margins, we have RMSEs of 0.0104 and 0.0037, respectively. Hence, calibration is still feasible without resorting to higher-powered computers, but the gains in fitting accuracy appear to be only moderate for the option price surface at hand.

Next, we drop the assumption of an equal mean reversion parameter and allow for $a_1\neq a_2$. Since the moment generating function of $(Y^1,Y^2)$ is then not known in closed form anymore, good starting values are particularly important in order to reduce computational time to an acceptable value. We distinguish the two cases where $\rho$ is diagonal (Step C) and $\rho$ is non-diagonal (Step D), and take as starting values, the parameters obtained from Step A or Step B, respectively. Interestingly, in \emph{Step C} the optimizer finds the minimum at the same parameters as in Step A, thus the additional freedom of different mean reversion parameters does not yield a better fit in this case. 

Finally, in \emph{Step D}, we calibrate the full model with non-diagonal $\rho$ and different mean reversion speeds $a_1,a_2$. Due to the lack of a closed-form expression for the moment generating function and the high number of parameters (15), the run-time increases to an unsatisfactory 10 hours on our standard PC, suggesting that higher-powered computing facilities and an optimized numerical implementation in a compiled instead of an interpreted language should be employed here. In contrast to Step C, we find an improvement by allowing for different mean reversion speeds: The overall RMSE is $0.0076$. Then again, for the data set at hand, the improvement is again only slight  compared to the simplest model considered in Step A. 

\begin{table}[ht]
\begin{footnotesize}
\begin{center}
\begin{tabular}{c||c|c|c|c|c|c|c|c|c|c|c|c|c|c|c}
 Step & $\lambda$ & $a_1$ & $\rho_1$ & $\rho_{12}$ & $\Theta^{11}$ & $\Sigma_0^{11}$ & $\gamma_1$ \\
 \hline 
 A & 0.774 & -2.392 & -3.741 & / & 0.011 & 0.019 & 0.027  \\ 
 B & 0.901 & -3.008 & -5.364 & 0.679 & 0.011 & 0.019 & 0.034 \\
 C & 0.774 & -2.392 & -3.741 & / & 0.011 & 0.019 & 0.027 \\
 D & 1.231 & -7.562 & -6.806 & 0.948 & 0.010 & 0.024 & 0.097 \\ 
 univ. 1 & 0.781 & -32.177 & -5.995 &/& 0.007 & 0.034& /  \\  
 univ. 2 & 0.864  & / & / & / & / & / & / \\
 initial & 0.800 & -2.500 & -3.000 & / & 0.010 & 0.020 & 0.020
\end{tabular}

\vspace{0.25cm}

\begin{tabular}{c||c|c|c|c|c|c|c|c|c|c|c|c|c|c|c}
 Step & $a_2$ & $\rho_2$ & $\rho_{21}$ & $\Theta^{22}$ & $\Sigma_0^{22}$ & $\gamma_2$ & $\Theta^{12}$ & $\Sigma_0^{12}$ \\
 \hline 
 A & / & -0.494 & /& 0.063 & 0.017 & 0.000 & 0.022 & 0.013  \\ 
 B & / & -0.661 & 0.896 & 0.067 & 0.018 & 0.000 & 0.023 & 0.013 \\
 C & -2.392 & -0.494 & /& 0.063 & 0.017 & 0.000 & 0.022 & 0.013  \\
 D  & -6.553 & -0.535 & 1.188 & 0.102 & 0.021 & 0.000 & 0.030 & 0.016 \\ 
 univ. 1   & / & / & / & / & / & / & / & / \\  
 univ. 2 & -2.482 & -0.471 & / & 0.050 & 0.017 & 0.012 & / & / \\
 initial &/ & -0.500 & / &  0.030 & 0.015 & 0.011 & 0.010 & 0.010 
\end{tabular}
\caption{Calibrated parameters for different models. In decreasing order: models from step A to D; univariate BNS model for EUR/USD and GBP/USD; initial parameters.}
\label{table: par}
\end{center}
\end{footnotesize}
\end{table}

\paragraph{Comparison with other bivariate models}

We now compare our bivariate Wishart-OU model to some benchmarks from the literature. The canonical candidate would be the bivariate Wishart model, which also exhibits stochastic correlations between the assets and has very recently been calibrated to market prices by \cite{dafonseca.grasselli.10}. However, the involved parameter restrictions necessary for the existence of the Wishart process are not satisfied in the results of the calibration. This suggests that some kind of constrained optimization must be incorporated, which is beyond our scope here. However, we emphasize that the Wishart model should yield a comparable performance once these implementation issues have been resolved in a satisfactory manner.

Instead, we use the multivariate Variance Gamma (henceforth VG) model of \cite{Luciano2006}, and a generalization with stochastic volatility suggested therein for our comparison.  In the mutivariate VG model with parameters $(\theta_i,\sigma_i,\nu)$, $i=1,2$, the log-price processes $Y^1,Y^2$ are given by two independent Brownian motions with drift which are subordinated by a common Gamma process. The joint moment generating function of the log-price processes under a risk neutral measure is shown to be given by
\begin{align*}
E[\exp(y_1 Y^1_t+y_2 Y^2_t)]= e^{(y_1(r_{\$}-r_{\eur}+w_1)+y_2(r_{\$}-r_{\pounds}+w_2))t}\left(1-\nu\sum_{i=1}^2{\left(y_i\theta_i+\frac12 y_i^2\sigma_i^2\right)}\right)^{-t/\nu},
\end{align*}
with $w_i=\nu^{-1}\log\left(1-\theta_i\nu-\frac12\sigma_i^2\nu\right)$. The parameters obtained from a calibration of this model to our option data set can be found in Table \ref{table: par2}. The corresponding overall RMSE is $0.0134$, which is roughly $63\%$ higher than the RMSE obtained from the calibration of our 12-parameter OU-Wishart model from Step A. At the EUR/USD and GBP/USD margin the multivariate VG model has a RMSE of $0.0161$ and $0.0107$. Consequently, the performance of this model is much worse than for the OU-Wishart model, which is not surprising since it only involves 5 parameters.

To alleviate this issue, our second benchmark allows for stochastic activity driven by an OU type process. More specifically, the log-price processes of the EUR/USD and GBP/USD spot rate are given by $Y^1_t={X_{Z_t}^1}$ and $Y^2_t={X_{Z_t}^2}$, where $X^1$ and $X^2$ are two independent Variance Gamma processes with parameters $(\theta_i,\sigma_i,\nu_i)$, $i=1,2$, and $Z_t=\int_0^t z_s ds$ is an integrated Ornstein-Uhlenbeck process. The Ornstein-Uhlenbeck process $(z_s)_{s\in\R^+}$ is given by $dz_s=2\alpha z_s ds+dN_{-2\alpha t}, z_0=1$, $\alpha<0$, where $N$ is a compound Poisson process with intensity $\vartheta$ and $\mathrm{Exp}(\xi)$ distributed jumps. It can be shown that the moment generating function of $Z_t$, see, e.g., \cite[7.2.2]{Schoutens2003}, is given by
\begin{align*}
 \Phi_{Z_t}(y)=\exp\left(\frac{y}{2 \alpha}(\exp(2 \alpha t)-1) + \frac{2 \alpha \vartheta (t y - \xi log[-2 \alpha \xi] + \xi log[(\exp(2 \alpha t)-1) y - 2 \alpha \xi])}{y + 2 \alpha \xi}\right).
\end{align*}
For the moment generating function of $Y_t=(Y^1_t,Y^2_t)$, conditioning on the stochastic activity process $Z$ yields
\begin{align*}
 \Phi_{Y_t}(y_1,y_2) = \Phi_{Z_t}\left( \log\Phi_{X_1^1}(y_1) + \log\Phi_{X_1^2}(y_2) \right)
\end{align*}
with $\Phi_{X_1^i}(y_i)= \left(1-y_i\theta_i\nu_i-\frac12\sigma_i^2 y_i^2\nu_i \right)^{-1/\nu_i}$, $i=1,2$. Thus, the joint moment generating function of the log-price processes $Y^1_t, Y^2_t$ under a risk neutral measure is given by 
\begin{align*}
 \Phi_{Y_t}(1,0)^{-y_1} \Phi_{Y_t}(0,1)^{-y_2} \Phi_{Y_t}(y_1,y_2).
\end{align*}

\begin{table}[ht]
\begin{footnotesize}
\begin{center}
\begin{tabular}{c|c|c|c|c|c|c|c|c}
 $\theta_1$ & $\theta_2$ & $\sigma_1$ & $\sigma_2$ & $\nu_1$ & $\nu_2$ & $\vartheta$ & $\alpha$ & $\xi$  \\
 \hline 
-0.360 & -0.327 & 0.090 & 0.093 & 0.106 & 0.106 & / & / & / \\
-1.470 &-2.190 &0.001 &0.050 &0.022 &0.001 &0.468 &-42.140 & 1.747
\end{tabular}
\caption{The first row shows the calibrated parameters for the multivariate VG model of \cite{Luciano2006}. The second row contains the calibrated parameters for two independent VG processes with a common integrated $\Gamma$-OU time change.}
\label{table: par2}
\end{center}
\end{footnotesize}
\end{table}

A calibration of this model to our dataset leads to the parameters provided in Table \ref{table: par2}; a plot depicting some of the respective implied volatilities can be found in Figure \ref{fig:impvol3}. The corresponding RMSE is $0.0129$. Somewhat surprisingly, this is only around $4\%$ lower than for the model of \cite{Luciano2006}, despite increasing the parameters from  5 to 9. At the margins, we have $0.0143$ and $0.0095$, which corresponds to improvements of around 11\%. Hence, there is quite some improvement in fitting the margins, but the multivariate options are not fit much better. This suggests that stochastic correlations indeed seem necessary to recapture the features of our empirical dataset.  However, let us emphasize again that this only applies to one specific dataset in the foreign exchange market. A more detailed empirical study is a challenging topic for future research.

\begin{figure}[h] 
\includegraphics[width=\linewidth]{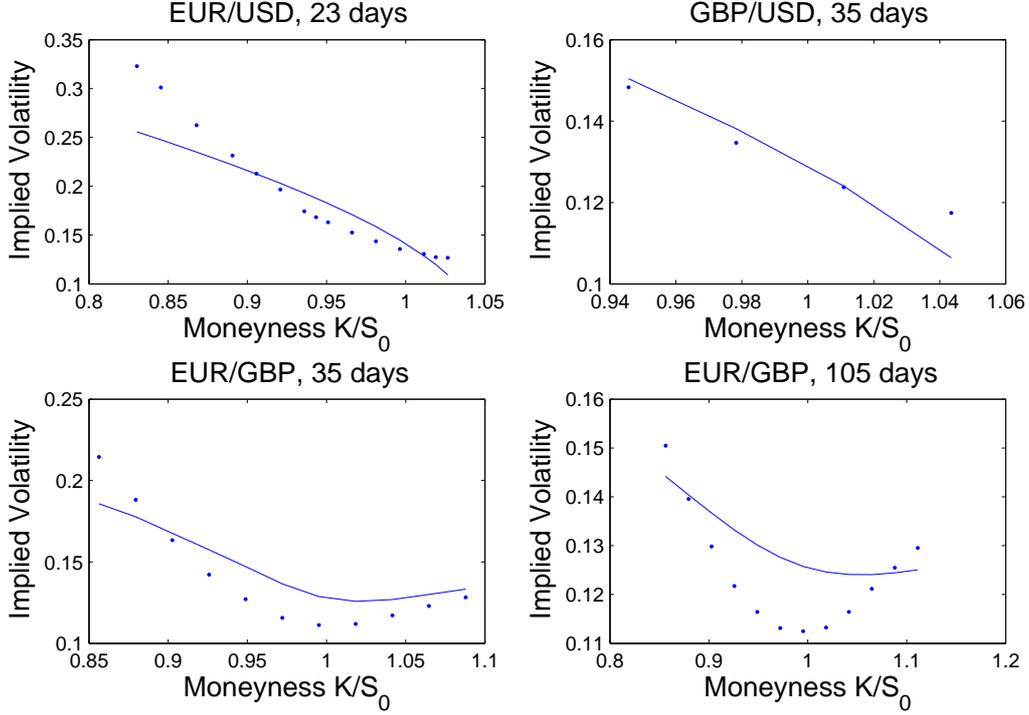}
\caption[]{Comparison of the Black-Scholes implied volatility of market prices (dots) and model prices (solid line). The headers state the underlying and the days to maturity. The plots are for the benchmark model where the log-price processes are modelled by two independent VG processes with a common time change which is given by an integrated $\Gamma$-OU process. The plots for the multivariate VG model from \cite{Luciano2006} look very similar.}
\label{fig:impvol3}
\end{figure}

\section{Covariance swaps}\label{sec: cov swaps}

In this final section, we show that it is possible to price swaps on the covariance between different assets in closed form. This serves two purposes.  On the one hand, options written on the realised covariance represent a family of payoffs that only make sense in models where covariances are modeled as stochastic processes rather than constants. On the other hand, the ensuing calculations exemplify once more the analytical tractability of the present framework.

We consider again our multivariate stochastic volatility model of OU type under an EMM $Q$. In addition, we suppose that the matrix subordinator $L$ is square integrable, i.e.,  $\int_{\{||X||>1\}}||X||^2\kappa_L(dX)<\infty$. The pricing of options written on the realised variance resp.\ the quadratic variation as its continuous-time limit have been studied extensively in the literature, cf., e.g., \cite{carr.lee.08} and the references therein. Since we have a nontrivial correlation structure in our model, one can also consider \emph{covariance swaps} on two assets $i,j\in\{1,\ldots,d\}$, i.e., contracts with payoff $[Y^i,Y^j]_T-K$ with \emph{covariance swap rate} $K=E([Y^i,Y^j]_T)$ (see, e.g., \cite{CarrMadan1999}, \cite{DaFonsecaGrasselliIelpo2008}, or \cite{Swischuk2004} for more background on these products). Now, we show how to compute the covariance swap rate. We have
\[ [Y^i,Y^j]_T = [Y^i,Y^j]_T^c + \sum_{s\leq T} \Delta Y_s^i \Delta Y_s^j = (\Sigma_T^+)^{ij} + \rho^i(X) \rho^j(X) \ast \mu_T^L(dX). \]
Since $\kappa_L(dX)dt$ is the compensator of $\mu^L$, this yields
\beq E([Y^i,Y^j]_T) = (E(\Sigma_T^+))^{ij} + T \int_{\S_d^+} \rho^i(X) \rho^j(X) \kappa_L(dX), \label{eq: gen covswap}\eeq
where $\Sigma_T^+$ was defined in Equation \eqref{prop_int_ou}. Note that by \cite[Proposition 2.4]{pigorsch} and since $|\rho^i(X)\rho^j(X)| \leq ||\rho||^2 ||X||^2$, our integrability assumption on $L$ implies that the expectation is finite. The first summand can be calculated as follows. By setting $y=0$ in Theorem \ref{cf}, we obtain the characteristic function of $\Sigma_t$. Differentiation yields
\[ E(\Sigma_T) = e^{AT} \Sigma_0 e^{A^\T T} + e^{AT}\Abf^{-1}(E(L_1))e^{A^\T T} - \Abf^{-1}(E(L_1)), \]
where $E(L_1)=\gamma_L+\int_{\S_d^+} X \,\kappa_L(dX)$. Using Equation \eqref{prop_int_ou}, we obtain
\[ E(\Sigma_T^+) = \Abf^{-1}(E(\Sigma_T)-TE(L_1)-\Sigma_0), \]
so we only need to know $E(L_1)$. The second summand in (\ref{eq: gen covswap}) can analogously be computed by differentiating the characteristic function of the matrix subordinator $L$.

In our OU-Wishart model, where $L$ is a compound Poisson matrix subordinator plus drift with $\W_d(n,\Theta)$-distributed jumps, we have by \cite[Theorem 3.3.15]{grupta} that
\[ E(L_1) = \gamma_L + \lambda n \Theta. \]
If $\rho$ is diagonal, the second term in (\ref{eq: gen covswap}) simplifies to
\[ T \rho_i \rho _j \int_{\S_d^+} X_{ii} X_{jj} \,\nu(dX) = T \rho_i \rho _j \lambda n \( 2 \Theta_{ij}^2 + n \Theta_{ii} \Theta_{jj} \), \]
again by \cite[Theorem 3.3.15]{grupta}. Thus we have a closed form expression for the covariance swap rate:
\begin{align*} 
K =& \left( \Abf^{-1} \left[ e^{AT}(\Sigma_0+\Abf^{-1}(\gamma_L + \lambda n \Theta))e^{A^\T T} - \Abf^{-1}(\gamma_L + \lambda n \Theta) - T(\gamma_L + \lambda n \Theta) - \Sigma_0 \right] \right)^{ij} \\
&+ T \rho_i \rho _j \lambda n \( 2 \Theta_{ij}^2 + n \Theta_{ii} \Theta_{jj} \).
\end{align*}
For example, in the $2$-dimensional OU-Wishart model from Section \ref{sec: 2dim OU-Wishart} we have, for $i=1$ and $j=2$,
\[ K = \frac{1}{a_1+a_2}\left[\(e^{(a_1+a_2)T}-1\)\(\Sigma_0^{12}+\frac{\lambda n \Theta_{12}}{a_1+a_2}\) - T \lambda n \Theta_{12}\right] + T \rho_1 \rho _2 \lambda n \( 2 \Theta_{12}^2 + n \Theta_{11} \Theta_{22} \). \]
As an illustration we provide, in Figure \ref{fig:covswap}, a plot of the normalized covariance swap rate measured in volaility points, i.e., $T\mapsto \sqrt{\tfrac{1}{T}E([Y^1,Y^2]_T)}$, for our calibrated 12-parameter OU-Wishart model from Section \ref{ss:empirical} (Step A). 

\begin{figure}[ht] 
\includegraphics[width=\linewidth]{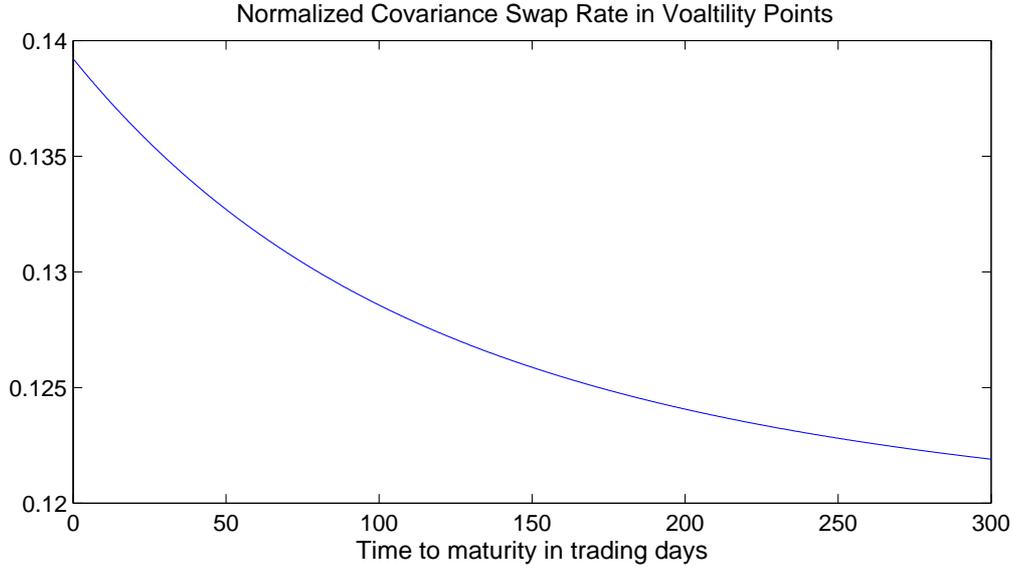}
\caption[]{Normalized covariance swap rate for the calibrated 12-parameter OU-Wishart model.}
\label{fig:covswap}
\end{figure}

Finally, we remark that similarly as in \cite{carr.lee.08}, pricing of options on the covariance can be dealt with using the Fourier methods from Section \ref{pricing}, since the joint characteristic function of $(\Sigma^+,\rho^i(X)\rho^j(X) \ast \mu^L(dX))$ can be calculated similarly as in the proof of Theorem \ref{cf}.

\begin{appendix}
\section{Appendix}
The following result on multidimensional analytic functions is needed in the proof of Lemma \ref{lemma_theta_L_analytic}.
 
\begin{lemma}\label{lemma_analytic_log}
Let $D_\epsilon=\{z\in\C^n: \norm{\Re(z)}<\epsilon\}$ for some $\epsilon>0$. Suppose $f:D_\epsilon\ra\C$ is an analytic function of the form $f=e^F$, where $F:D_\epsilon\ra\C$ is continuous. Then $F$ is analytic in $D_\epsilon$.
\end{lemma}

\begin{proof}
Let $z=(z_1,z_2,\ldots,z_n) \in D_{\epsilon}$ and define $z_{-1}=(z_2,\dots,z_n)$. Then $f_{z_{-1}}: w \mapsto f(w,z_{-1})$ defines an analytic function without zeros on the open convex set $D_{ \epsilon,z_{-1}} := \{ w\in\C: (w,z_{-1}) \in D_\epsilon \}$. By, e.g., \cite[Satz V.1.4]{Fischer1994},  there exists an analytic function $g_{z_{-1}}^1:D_{\epsilon,z_{-1}}\ra\C$ such that $\exp(g_{z_{-1}}^1)=f_{z_{-1}}$. Hence 
$F(w,z_{-1}) - g_{z_{-1}}^1(w) \in 2\pi i \Z$ on $D_{\epsilon,z_{-1}}$. Since both $F$ and $g$ are continuous, their difference is constant and it follows that $w \mapsto F(w,z_{-1})$ is analytic on $D_{\epsilon,z_{-1}}$. Analogously, one shows analyticity of $F$ in all other components. The assertion then follows from Hartog's Theorem (cf.,  e.g., \cite[Theorem 2.2.8]{Hormander1967}).
\end{proof}

\end{appendix}

\begin{small}
\bibliography{references}
\end{small}
\end{document}